\newif\ifarxiv

\documentclass[dvipsnames,11pt]{article}
\usepackage{color,amsthm,mathrsfs,extarrows,tikz,graphicx,enumitem,fancybox,makecell}
\setlist[itemize]{itemsep=0pt}
\setlist[enumerate]{itemsep=0pt}
\usepackage[pagebackref]{hyperref}
\usepackage[capitalise,nameinlink]{cleveref}
\usepackage[margin=1.in]{geometry}

\usepackage{amsfonts,amsmath,amssymb, mathtools}
\usepackage{bm}
\usepackage[square,sort,comma,numbers]{natbib}
\usepackage[ruled]{algorithm2e}
\usepackage{xcolor,bbm}
\usepackage{comment}
\usepackage{multirow}
\usepackage{tabu}  
\usepackage[textsize=scriptsize]{todonotes}

\crefname{subsection}{Section}{Sections}
\Crefname{subsection}{Section}{Sections}
\crefname{subsubsection}{Section}{Sections}
\Crefname{subsubsection}{Section}{Sections}

\allowdisplaybreaks

\hypersetup{colorlinks=true,urlcolor=blue,linkcolor=blue,citecolor=[rgb]{.42,.56,.14},}

\SetKwProg{Procedure}{Procedure}{:}{end}
\SetKwFunction{Sparse}{Sparse}
\SetKwFunction{Reduction}{Reduction}
\SetKwFunction{Classification}{Classification}
\SetKwFunction{Estimation}{Estimation}

\Crefname{lemma}{Lemma}{Lemmas}
\Crefname{fact}{Fact}{Facts}
\Crefname{theorem}{Theorem}{Theorems}
\Crefname{corollary}{Corollary}{Corollaries}
\Crefname{claim}{Claim}{Claims}
\Crefname{example}{Example}{Examples}
\Crefname{problem}{Problem}{Problems}
\Crefname{definition}{Definition}{Definitions}
\Crefname{notation}{Notation}{Notations}
\Crefname{assumption}{Assumption}{Assumptions}
\Crefname{subsection}{Subsection}{Subsections}
\Crefname{section}{Section}{Sections}
\Crefname{algorithm}{Algorithm}{Algorithms}
\Crefname{algocf}{Algorithm}{Algorithms}
\Crefformat{equation}{(#2#1#3)}

\newtheorem{theorem}{Theorem}[section]

\newtheorem{lemma}[theorem]{Lemma}
\newtheorem{corollary}[theorem]{Corollary}

\theoremstyle{definition}
\newtheorem{definition}[theorem]{Definition}

\newtheorem{exercise-easy}[theorem]{Exercise}
\newtheorem{exercise-med}[theorem]{Exercise}
\newtheorem{exercise-hard}[theorem]{Exercise$^\star$}

\newtheorem{remark}[theorem]{Remark}

\newtheorem{problem}[theorem]{Problem}
\newtheorem{fact}[theorem]{Fact}

\DeclareMathOperator*{\E}{\mathbb E}

\renewcommand{\Pr}{\operatorname*{\mathbf{Pr}}}

\newcommand{\eps}{\varepsilon}
\newcommand{\abs}[1]{\left| #1 \right|}
\newcommand{\vabs}[1]{\left\| #1 \right\|}

\newcommand{\pbra}[1]{\left( #1 \right)}
\newcommand{\sbra}[1]{\left[ #1 \right]}
\newcommand{\cbra}[1]{\left\{ #1 \right\}}
\newcommand{\ceilbra}[1]{\left\lceil #1 \right\rceil}

\renewcommand{\mid}{\,\middle\vert\,}
\newcommand{\bin}{\{0,1\}}

\newcommand{\binBT}{\{\bot,\top\}}

\newcommand{\TruncLap}{\mathsf{TruncLap}}

\newcommand{\poly}{\mathsf{poly}}

\newcommand{\taumin}{\tau_\textsf{min}}
\newcommand{\taumax}{\tau_\textsf{max}}

\newcommand{\indicator}{\mathsf{1}}
\newcommand{\Dec}{\mathsf{Dec}}
\newcommand{\Lap}{\mathsf{Lap}}
\newcommand{\Normal}{\mathsf{N}}
\newcommand{\Median}{\mathsf{Median}}
\newcommand{\bone}{{\bm1}}
\newcommand{\bzero}{{\bm0}}
\newcommand{\leaves}{\mathsf{leaves}}
\newcommand{\nodes}{\mathsf{nodes}}
\newcommand{\depth}{\mathsf{depth}}
\newcommand{\RMSE}{\mathsf{RMSE}}
\newcommand{\mRMSE}{\mathsf{mRMSE}}
\newcommand{\REL}{\mathsf{REL}}
\newcommand{\errltsq}{\ell_2^2\text{-}\mathsf{error}}
\newcommand{\errlinf}{\ell_\infty\text{-}\mathsf{error}}

\newcommand{\linflone}{\infty, 1}
\newcommand{\linfltwo}{\infty, 2}
\newcommand{\linflp}{\infty, p}

\newcommand{\Nbb}{\mathbb{N}}
\newcommand{\Rbb}{\mathbb{R}}
\newcommand{\Zbb}{\mathbb{Z}}

\newcommand{\Acal}{\mathcal{A}}
\newcommand{\Bcal}{\mathcal{B}}

\newcommand{\Ecal}{\mathcal{E}}
\newcommand{\Fcal}{\mathcal{F}}
\newcommand{\Hcal}{\mathcal{H}}

\newcommand{\Mcal}{\mathcal{M}}

\newcommand{\Pcal}{\mathcal{P}}
\newcommand{\Scal}{\mathcal{S}}
\newcommand{\Tcal}{\mathcal{T}}
\newcommand{\cT}{\Tcal}

\newcommand{\cM}{\mathcal{M}}
\newcommand{\R}{\mathbb{R}}
\newcommand{\Z}{\mathbb{Z}}
\newcommand{\N}{\mathbb{N}}
\newcommand{\bW}{{\bm W}}
\newcommand{\bU}{{\bm U}}
\newcommand{\bR}{{\bm R}}
\newcommand{\bA}{{\bm A}}
\newcommand{\bx}{{\bm x}}
\newcommand{\bu}{{\bm u}}
\newcommand{\bv}{{\bm v}}
\newcommand{\bz}{{\bm z}}
\newcommand{\ind}{{\mathbb{I}}}

\newcommand{\tz}{\tilde{z}}
\newcommand{\tw}{\tilde{w}}
\newcommand{\bw}{{\bm w}}
\newcommand{\tbw}{\tilde{\bw}}

\providecommand{\Comments}{0} 
\ifnum\Comments=1
\paperwidth=\dimexpr \paperwidth + 2.8cm\relax
\evensidemargin=\dimexpr\evensidemargin + 2.2cm\relax
\fi
\newcommand{\mytodo}[1]{\ifnum\Comments=1{#1}\fi}

\renewcommand{\tilde}{\widetilde}

\title{On Differentially Private Counting on Trees}

\author{
Badih Ghazi\thanks{Google, Mountain View, CA. Email: \texttt{\{badihghazi, ravi.k53\}@gmail.com, pritish@alum.mit.edu}.}
\and
Pritish Kamath\footnotemark[1]
\and
Ravi Kumar\footnotemark[1]
\and
Pasin Manurangsi\thanks{Google, Thailand. Email: \texttt{pasin@google.com}.}
\and
Kewen Wu\thanks{University of California, Berkeley.  Email: \texttt{shlw\_kevin@hotmail.com}. Most of this work was done while at Google.}
}
\date{}

\begin{document}
\maketitle

\begin{abstract}
We study the problem of performing counting queries at different levels in hierarchical structures while preserving individuals' privacy.  Motivated by applications, we propose a new error measure for this problem by considering a combination of multiplicative and additive approximation to the query results.  
We examine known mechanisms in differential privacy (DP) and prove their optimality, under this measure, in the pure-DP setting. In the approximate-DP setting, we design new algorithms achieving significant improvements over known ones.
\end{abstract}

\thispagestyle{empty}
\setcounter{page}{0}
\newpage

 \section{Introduction}\label{sec:introduction}

With the increasing need to preserve the privacy of users, differential privacy (DP) \cite{DBLP:journals/jpc/DworkMNS16,DBLP:conf/eurocrypt/DworkKMMN06} has emerged as a widely popular notion that provides strong guarantees on user privacy and satisfies compelling mathematical properties. There have been many deployments of DP in the field of data analytics both in industry \cite{dp2017learning, ding2017collecting} and by government agencies \cite{abowd2019economic}.

We start by recalling the formal definition of DP, tailored to our setting. 
\begin{definition}[Differential Privacy]\label{def:differential_private}
Let $\Acal$ be a randomized algorithm taking an integer vector as input.   We say $\Acal$ is \emph{$(\eps,\delta)$-differentially private} (i.e., \emph{$(\eps,\delta)$-DP}) if 
$$
\Pr\sbra{\Acal(\bx)\in S}\le e^\eps\cdot\Pr\sbra{\Acal(\bx')\in S}+\delta,
$$
holds for any measurable subset $S$ of $\Acal$'s range and any two neighboring inputs $\bx, \bx'$, where $\bx, \bx'$ are considered \emph{neighbors} iff $\vabs{\bx - \bx'}_1 = 1$.

When $\delta=0$, we say $\Acal$ is \emph{$\eps$-DP} (aka \emph{pure-DP}); the case $\delta > 0$ is \emph{approximate-DP}.
\end{definition}

\paragraph*{Estimating Counts in Trees.}
A fundamental task in data analytics is to aggregate counts over hierarchical subsets (specifically, trees) of the input points. For example, the government might be interested in the number of households, aggregated at the state, country, and city levels. As another example, online advertisers might be interested in the number of user clicks on product ads, when there is a category hierarchy on the products. 
The tree aggregation problem has been the subject of several previous works in DP including in the context of range queries~\cite{CormodePSSY12,XiaoXFGL14,DworkNRR15,ZhangXX16}, the continuous release model~\cite{DworkNPR10,ChanSS11}, private machine learning \cite{DBLP:conf/icml/KairouzM00TX21,537076}, and the US census top-down algorithms~\cite{census-arxiv,census-main,Cohen2022Private,CohenDMS21}, to name a few\footnote{We remark that there is a reduction from our problem to that of releasing thresholds, which we discuss in more detail in \cref{sec:threshold}.}. 
In this work, we revisit this basic problem and present new perspectives and results.

Let $\Tcal$ be a rooted tree of depth\footnote{The \emph{depth} is defined to be the maximum number of nodes along a root-to-leaf path of the tree.} $d$ and arity $k$; the structure of $\Tcal$ is known a priori.
Let $\nodes(\cT)$ be the set of \emph{nodes} and $\leaves(\cT)$ be the set of \emph{leaves} in $\cT$.  The problem of private aggregation in trees can be formalized as follows. 
\begin{problem}[Tree Aggregation] \label{prob:tree_counting_problem}
Given a tree $\Tcal$, 
the input to the problem is a vector
$\bx \in \Nbb^{\leaves(\Tcal)}$, where
$x_v\in\Nbb$ is a value for  $v\in\leaves(\Tcal)$. For each node $u\in\Tcal$, define its weight $w_u$ by
$$
w_u=\sum_{v \text{\rm{ is a leaf under }}u} x_v.
$$
The desired output is a DP estimate vector $\tbw \in \Rbb^{\nodes(\Tcal)}$ of $\bw$.
\end{problem}

In the above formulation, the input $x_v$ represents the number of individuals that contribute to the leaf $v$, and the weight $w_u$ counts all the number of individuals that contribute to any of its descendants (or itself).  As before, 
$\bx, \bx'$ are neighbors iff $\vabs{\bx - \bx'}_1 = 1$.

Besides being a natural problem on its own, algorithms for tree aggregation also serve as subroutines for solving other problems such as range queries~\cite{CormodePSSY12,XiaoXFGL14,DworkNRR15,ZhangXX16}.

\paragraph*{Linear Queries and Error Measure.}
Tree aggregation in fact belongs to a class of problems called \emph{linear queries} --- one of the most widely studied problems in DP (see, e.g., \cite{DinurN03,DworkMT07,HardtT10,BhaskaraDKT12,NikolovTZ13,BunUV18,Nikolov15,BlasiokBNS19,EdmondsNU20,Nikolov-JL}). In its most general form, the problem can be stated as follows.

\begin{problem}[Linear Queries]\label{prob:linear_queries}
For a given workload matrix $\bW \in \R^{m \times n}$, the input to the \emph{$\bW$-linear query} problem is a vector $\bx \in \Nbb^n$ and the  output is a DP estimate of $\bW\bx$.
\end{problem}

It is easy to see that the tree aggregation problem can be viewed as a linear query problem, where the binary workload matrix $\bW^{\cT} \in \{0, 1\}^{\nodes(\cT) \times \leaves(\cT)}$ encodes if each leaf (corresponding to a column index) is a descendant of (or itself) each node (corresponding to a row index).

Two error measures have been studied in the literature: 
the (expected) \emph{$\ell^2_2$-error}\footnote{In some previous work, $\errltsq(\cM; \bW)$ is defined as $\max_{\bx \in \N^n}  \frac{1}{m} \E\sbra{\vabs{\cM(\bx) - \bW\bx}_2^2}$ (without the square root). We use the current version as it is more convenient to deal with in our error analysis. In any case, we can obviously convert a bound in one version to the other.}
\begin{align*}
\errltsq(\cM; \bW) := \max_{\bx \in \N^n}  \sqrt{\frac{1}{m} \E\sbra{\vabs{\cM(\bx) - \bW\bx}_2^2}},
\end{align*}
and the (expected) \emph{$\ell_\infty$-error}
\begin{align*}
\errlinf(\cM; \bW) := \max_{\bx \in \N^n} \E\sbra{\vabs{\cM(\bx) - \bW\bx}_{\infty}},
\end{align*}
where $\cM$ is a DP mechanism for the $\bW$-linear query problem.
Indeed, previous works have characterized the best possible errors in the approximate-DP case up to polylogarithmic factors for any given workload $\bW$. (See the discussion in \cite{EdmondsNU20} for more details.) 

It is worth noting that these measures focus only on the \emph{additive error} of the query, i.e., $\cM(\bx) - \bW\bx$. In many scenarios, however, this is not the only possible measure of error. Specifically, in this work, we seek to expand the error measure by additionally incorporating \emph{multiplicative error}. Intuitively, multiplicative errors are meaningful when the true answer (i.e., $(\bW\bx)_i$) is quite large; e.g., if the true error is $10^6$, then we should not be distinguishing whether the additive error is $10$ or $100$ as both of them are very small compared to $10^6$. 
In addition to this intuition, multiplicative errors have also been used in other contexts such as in empirical evaluations of range queries (e.g., \cite{CormodePSSY12,QardajiYL13,ZhangXX16}).

With the above discussion in mind, we now proceed to define the error measure.

\begin{definition}[Multiplicative Root Mean Squared Error] \label{def:err-rmse}
Given parameter $\alpha > 0$, we define an \emph{$\alpha$-multiplicative root mean squared error ($\alpha$-$\RMSE$)} of an estimate $\tz$ of the true answer $z\ge0$ as
\begin{align*}
\RMSE_{\alpha}(\tz, z) := \sqrt{\E_{\tz}\sbra{\pbra{\max\cbra{\abs{\tz - z} - \alpha\cdot z, 0}}^2}}.
\end{align*}

For $\bW$-linear query, we define an \emph{$\alpha$-multiplicative maximum root mean squared error ($\alpha$-$\mRMSE$)} of a mechanism $\Mcal$ to be
\begin{align*}
\mRMSE_{\alpha}(\cM; \bW) := \max_{\bx \in \Z^n} \max_{i \in [m]} \RMSE_{\alpha}\pbra{\cM(\bx)_i, (\bW\bx)_i}.
\end{align*}
\end{definition}

Note that when $\alpha = 0$ (i.e., the error is only additive), our notion of $\alpha$-$\RMSE$ coincides with that of the standard RMSE. 
By taking the maximum error across all queries when defining the error for linear queries, we mitigate the weakness of $\ell_2^2$-error bound, which allows some queries to incur huge errors, while still avoiding the ``union bound issue'' faced in the $\ell_\infty$-error. 
The latter can be significant as the number of queries here can be exponential in the depth $d$.

We remark that our algorithms also achieve the usual with high probability guarantees, i.e., with probability at most $\eta$, $|\tz - z| \leq \alpha\cdot \max\{z, \tau\}$ for some threshold $\tau$.
We defer such a statement to later sections for simplicity of comparing the bounds. Furthermore, our error notion implies upper bounds on ``smoothed relative errors'' used for empirical evaluations in previous works \cite{QardajiYL13,ZhangXX16}. We provide a formal statement in \Cref{app:err-comp}.

When $\alpha = 0$, we drop the ``$\alpha$-multiplicative'' or ``$\alpha$-'' prefixes and refer to the errors simply as maximum RMSE or $\mRMSE$. Similarly, we also drop $\alpha$ from the subscript and simply write $\mRMSE$ instead of $\mRMSE_0$.

\subsection{Our Results}\label{sec:our_results}

\begin{table}[ht]
	\renewcommand\arraystretch{1.5}
	\centering
\scalebox{0.95}{%
	\begin{tabu}{|[1.2pt]l|[1.2pt]cl|cl|[1.2pt]}
		\tabucline[1.2pt]{-}
		{\bf Type of error} & \multicolumn{2}{c|}{\bf\boldmath$\eps$-DP} & \multicolumn{2}{c|[1.2pt]}{\bf\boldmath$(\eps, \delta)$-DP}\\
		\tabucline[1.2pt]{-}
		\multirow{2}{*}{Additive-only ($\alpha = 0$)} &
		$O(d/\eps)$ & :~Laplace &
		$O_{\eps,\delta}(\sqrt{d})$ & :~Gaussian \\
		\cline{2-5}
		&
		$\Omega(d/\eps)$ & :~\Cref{thm:lower_bound} &
		$\Omega_{\eps,\delta}(\sqrt{d})$ & :~\Cref{thm:add-err-lb}\\
		\hline
		Additive-Multiplicative ($0 < \alpha < 1$) &
		$\Omega(d/\eps)$ & :~\Cref{thm:lower_bound} &
		$O_{\eps,\delta}(\log d)$ & :~\Cref{thm:upper_bound}\\
		\tabucline[1.2pt]{-}
	\end{tabu}}
	\caption{Overview of results; entries indicate upper/lower bounds on $\alpha$-$\mRMSE$. Upper bounds corresponding to Laplace and Gaussian mechanisms are formally stated in \Cref{cor:baseline_algorithms}. For simplicity we omit the dependence on $\alpha$ in the additive-multiplicative bounds here.}
	\label{tab:overview}
\end{table}

Two known baselines for tree aggregation are the $\eps$-DP Laplace mechanism and $(\eps, \delta)$-DP Gaussian mechanism, which achieve $\mRMSE$ of $O(d/\eps)$ and $O(\sqrt{d \log (1/\delta)} / \eps)$ respectively. We start by showing that these are already tight for the additive-only errors:

\begin{theorem}[Informal; see \Cref{thm:lower_bound}]\label{thm:thm:lower_bound_informal_1}
There is no $\eps$-DP algorithm for tree aggregation with $\mRMSE$ $o(d/\eps)$, even for binary trees.
\end{theorem}

\begin{theorem}[Informal; see \Cref{thm:add-err-lb}] \label{thm:add-err-apx-dp-informal}
There is no $(\eps, \delta)$-DP algorithm for tree aggregation with $\mRMSE$ $o_{\eps, \delta}(\sqrt{d})$, even for binary trees.
\end{theorem}

Given the above results, it is therefore natural to ask whether multiplicative errors can help reduce the error bound. For pure-DP, we show that this unfortunately is not the case.

\begin{theorem}[Informal; see \Cref{thm:lower_bound}]\label{thm:thm:lower_bound_informal_2}
For any constant $\alpha < 1$, there is no $\eps$-DP algorithm for tree aggregation with $\alpha$-$\mRMSE$ $o(d/\eps)$, even for binary trees.
\end{theorem}

Our next --- and perhaps the most surprising --- result is that, unlike in the pure-DP case, allowing multiplicative approximation in approximate-DP allows us to reduce the upper bounds exponentially from $O_{\eps, \delta}(\sqrt{d})$ to $O_{\eps, \delta}(\log d)$:

\begin{theorem}[Informal; see \Cref{thm:upper_bound}]\label{thm:upper_bound_informal}
For any constant $\alpha > 0$, there is an efficient $(\eps, \delta)$-DP algorithm for tree aggregation with $\alpha$-$\mRMSE$ $O(\log(d/\delta) / \eps)$.
\end{theorem}

We remark that \Cref{thm:upper_bound_informal} has worse dependency on $\delta$ than the $(\eps,\delta)$-DP Gaussian mechanism. 
Indeed, the former has $\log(1/\delta)$ whereas the latter only has $\sqrt{\log(1/\delta)}$.
However this gap is somewhat unavoidable as we will discuss in \Cref{rmk:optimality_apx} when $\delta$ is small, say, $\delta=2^{-\Omega(d)}$.
Our results are summarized in \Cref{tab:overview}.

Our bounds do \emph{not} depend on the arity $k$ of $\Tcal$. 
This is immediate for the lower bounds (\Cref{thm:thm:lower_bound_informal_1,thm:add-err-apx-dp-informal,thm:thm:lower_bound_informal_2}) since it suffices to prove it for binary trees $k=2$.
The reason for the upper bounds (\Cref{thm:upper_bound_informal}) is less clear, relying on the fact that the weights of two nodes are correlated iff they are on the same root-to-leaf path, which is irrelevant of the arity. 

\subsection{Proof Overview}\label{sec:proof_overview}

\paragraph*{Probabilistic Utility Guarantee.}
Recall that we defined the error $\alpha$-$\mRMSE$ as a variant of RMSE but with a multiplicative error subtracted out.
While this gives us a nice scalar quantity (once $\alpha$ is fixed) to work with and state the results, it will be useful in the subsequent analyses to define a probabilistic version of the guarantee with additional fixed thresholds.

In this different utility guarantee (formalized in~\Cref{eq:thresh}), every node $u$ is given an additional threshold $\tau_u$, and a randomized output $\tbw$ is accurate if for all $u\in\nodes(\Tcal)$ we have with probability at least $1-\eta$ that
$$
\abs{\tw_u-w_u}\le\alpha\cdot\max\cbra{w_u,\tau_u}.
$$
The smaller the thresholds $\tau_u$'s are, the better the accuracy will be.

The benefit of having $\tau_u$ is that it is independent of $w_u$ and is explicitly available to the algorithms; this formulation may be of independent interest.
On the other hand, this probabilistic guarantee is closely related to the original $\alpha$-$\mRMSE$ in \Cref{def:err-rmse}:
\begin{enumerate}
\item\label{itm:proof_overview_1} Any algorithm with low $\alpha$-$\mRMSE$ has good probabilistic guarantee (\Cref{lem:from_mrmse}).
\item\label{itm:proof_overview_2} Any algorithm with good probabilistic guarantee can be converted into one with low $\alpha$-$\mRMSE$ (\Cref{lem:to_mrmse}).
\end{enumerate}

\paragraph*{Improved Approximate-DP Algorithm.}
Given \Cref{itm:proof_overview_2}, it suffices to design an algorithm with a good probabilistic guarantee, i.e., works for thresholds $\tau_u$'s as small as possible.
Our algorithm can be decomposed into two parts: a reduction step and a classification algorithm.

\smallskip
\textbf{Reduction.} 
Since the error measure is multiplicative when $w_u$ is large, we design a geometric sequence of thresholds and classify each $w_u$ into the correct interval created by the thresholds.
To do so, every time we use a classification algorithm to find nodes that are above the current threshold, and in the next round we only focus on the ones below the threshold.
Assuming previous classifications are all correct, the weights of the nodes above the current threshold are actually below the previous threshold. Such a sandwiching relation provides a good approximation if the granularity of the thresholds is not too large compared with the ratio $\alpha$ (\Cref{lem:reduction}).

Moreover, we assign privacy and error parameters in the same geometric fashion, thus their telescoping sum (from composition theorems) converges (\Cref{thm:estimation}).

\smallskip
\textbf{Classification.}
Given any fixed threshold $\tau$, the goal is to correctly classify each $w_u$ to be either \emph{above} or \emph{below} $\tau$.
Naively, to ensure every node is correctly classified, we need to apply a union bound over all the nodes. This will incur a $\poly(d)$ overhead if we use Laplace noise or Gaussian noise as in the standard mechanisms (\Cref{lem:laplace_mechanism,lem:gaussian_mechanism}) since the tree can have exponential size.
To deal with this issue, we use a \emph{truncated Laplace mechanism} where the Laplace noise is truncated to be bounded (\Cref{lem:truncated_laplace_mechanism}); this ensures that the estimation error is always at most the truncation range and thus can obviate a union bound. 
However, we still need to pay the privacy loss from compositions.
If one node is the ancestor of another, then the input leaves they depend on must overlap, which means simply estimating every node's weight will incur $d$ rounds of composition.
To improve this, our classification algorithm will find the transition nodes in the tree: a \emph{transition node} is one whose weight exceeds $\tau$ but none of its children has weight above $\tau$.
Given the locations of the transition nodes, we can easily classify the other nodes: a node is above $\tau$ iff it is the ancestor of (or itself) a transition node.
Assuming the previous classification with threshold $\tau'>\tau$ succeeds, none of the nodes'  weights should exceed $\tau'$ now.
Since there are at most $\tau'/\tau$ transition nodes in a tree, we can use the \emph{sparse vector technique} \cite{DBLP:conf/stoc/DworkNRRV09} to find them, and incur fewer rounds of composition.

\smallskip
We remark that the idea of using increasing thresholds and sparse vector techniques has been used in \cite{bolot2013private,DBLP:conf/esa/FichtenbergerHO21,epasto2023differentially}.

\paragraph*{Pure-DP Lower Bounds.}
Given \Cref{itm:proof_overview_1}, it suffices to rule out DP algorithms with very strong probabilistic guarantees, i.e., works for thresholds $\tau_u$'s that are too small.

Our proof uses the \emph{packing argument} \cite{HardtT10}: 
we construct extremal datasets where any two datasets have a large distance.
Then if the output has small error, we can correctly identify the input dataset.
On the other hand, by the privacy guarantee, the output distribution for different datasets, though having a large distance, should not be too different, contradicting the fact that they decode to different input datasets.

Not surprisingly, our extremal datasets place the maximal value on a leaf and keep other leaves empty.
But the key issue is the decoding step.
Indeed, previous packing arguments work with $\errlinf$, where the error on \emph{all} output coordinates is small with high probability, thus admitting simple decoding algorithms.
We, however, can only guarantee the error on any \emph{fixed} node is small with high probability; we also cannot use a union bound since the output size can be exponential.

The way we circumvent this is by designing a novel probabilistic decoding algorithm where we will correctly decode to the input dataset with probability large enough to derive a contradiction.
The decoding algorithm itself performs a random walk on the tree where each step favors the larger estimated weight. 
Then the success probability of decoding can be lower bounded in terms of the number of correctly classified nodes, which in turn can be lower bounded by its expectation and our probabilistic guarantee suffices.

\paragraph*{Additive-Only Lower Bounds.}
The above packing argument only works for pure-DP setting (or $(\eps,\delta)$-DP but with exponentially small $\delta$).
Indeed, as shown by our improved approximate-DP algorithm (\Cref{thm:upper_bound_informal}), the bound can be exponentially small if we allow both approximate-DP and $\alpha>0$.
Therefore we now turn to the only remaining case: approximate-DP and $\alpha=0$. 
In this case, by \Cref{def:err-rmse}, $\alpha$-$\mRMSE$ is an additive-only error.

Our proof starts by slightly modifying the error characterization of linear queries from \cite{EdmondsNU20}. This shows that $\mRMSE$ for $\bW$-linear query is characterized by a factorization norm of $\bW$. 
Thus proving \Cref{thm:add-err-apx-dp-informal} boils down to showing a lower bound on this factorization norm of the binary tree matrix. 
Following previous works on range queries (e.g., \cite{rny033}), we do so by invoking a dual (maximum-based) characterization of the factorization norm from \cite{mathias1993hadamard,LeeSS08} and give an explicit solution to this dual formulation. 

\subsection{Relation Between Tree Aggregation and Releasing Thresholds}
\label{sec:threshold}

There is a simple reduction from the tree aggregation problem (\Cref{prob:estimation}) to the problem of releasing thresholds. 
Recall that the problem of releasing thresholds is the same as linear queries with the workload matrix $\bW^{\mathrm{th}} \in \{0, 1\}^{m \times m}$ being the matrix with upper-triangular entries (including the diagonal entries) equal to one.

The reduction works as follows. First, we may index the leaves from $1, \dots, |\leaves(\cT)|$ based, say, on their order in the DFS traversal of the tree. It is not hard to see that $w_u$ of any node in the tree corresponds to $(\bW^\mathrm{th}\bx)_b - (\bW^\mathrm{th}\bx)_{a - 1} = \sum_{v \in [a, b]} x_v$ for some $a, b \in \{1, \dots, |\leaves(\cT)|\}$. Therefore, we may run any DP threshold releasing algorithm (with $m = |\leaves(\cT)|$) and solve the tree aggregation problem.

The above reduction yields an mRMSE error for the tree aggregation problem that is of the same order as that of releasing thresholds (with $m = |\leaves(\cT)|$). For the latter, it is known that the tight error is $\Theta_{\eps, \delta}(\log m)$~\cite{Henzinger-continual-counting}\footnote{In fact, the lower bound in \cite{Henzinger-continual-counting} is even stronger as it holds against the $\ell_2^2$-error.}. Assuming that each node at depth less than $d$ has at least two children, $|\leaves(\cT)| \geq 2^d$ and therefore, the error yielded by this reduction is at least $\Omega_{\eps, \delta}(d)$. In other words, this is not even as good as the straightforward Gaussian mechanism for our problem, which yields an error of $O_{\eps, \delta}(\sqrt{d})$ (and we have shown this to be tight in \Cref{thm:add-err-lb}).

We note that there is another line of research that studies privately learning threshold functions. Recent work has shown that learning threshold functions with $(\eps, \delta)$-DP in the PAC model only requires $O_{\alpha, \eps, \delta}\pbra{(\log^* m)^{O(1)}}$ samples~\cite{KaplanThreshold}. Due to the connection between learning thresholds and releasing threshold functions presented in~\cite{BunThreshold}, this gives an algorithm for the latter with an error bound of $O_{\eps, \delta}\pbra{(\log^* m)^{O(1)} \cdot (\log n)^{2.5}}$, where $n$ denotes $\|\bx\|_1$ (i.e., the total count across all leaves in our setting).
When combining this bound with the above reduction, one gets a tree aggregation algorithm with error $O_{\eps, \delta}\pbra{(\log^* |\leaves(\cT)|)^{O(1)} \cdot (\log n)^{2.5}}$. Although the term $(\log^* |\leaves(\cT)|)^{O(1)}$ is very small, this error bound is not directly comparable to the lower and upper bounds achieved in our paper because our bounds are \emph{independent} of $n$ whereas there is a dependency of $(\log n)^{2.5}$ in this releasing threshold-based bound. 
(Note also that the dependency on $n$ cannot be removed while keeping the dependency on $m$ sublogarithmic, as this would contradict with the aforementioned lower bound of~\cite{Henzinger-continual-counting}.)

Finally, we remark that the reduction does \emph{not} work if we allow the threshold releasing algorithm to incur multiplicative errors. This is because we need to subtract two thresholds to get $w_u$ for each node $u$ in the tree, and subtraction does not preserve multiplicative approximation guarantees.

\subsection*{Paper Organization.}
We formalize the notation in \Cref{sec:preliminary}.
Then in \Cref{sec:util_formulation} we introduce the error measure we will actually use in designing algorithms and proving lower bounds; and relate it to the $\alpha$-$\mRMSE$ measure.
The improved approximate-DP algorithm is presented in \Cref{sec:upper_bounds} and the corresponding lower bounds are in \Cref{sec:lower_bounds}.
The concluding remarks are in \Cref{sec:conclusion}.
\section{Preliminaries}\label{sec:preliminary}

We use $\ln(\cdot)$ and $\log(\cdot)$ to denote the logarithm with base $e$ and $2$ respectively. 
For a positive integer $n$, let $[n]$  denote the set $\cbra{1,\ldots,n}$.
Let $\Nbb$ denote the set of non-negative integers.

\subsection{Norms}

We use boldface uppercase (e.g., $\bA$) to denote matrices and boldface lowercase (e.g., $\bx$) to denote vectors. We use $\bzero, \bone$ to denote the all-zeros and all-ones vectors / matrices.

For $\bx \in \R^m$, its 
\emph{$\ell_p$-norm} is defined as $\vabs{\bx}_p := \pbra{\sum_{i \in [m]}\abs{x_i}^p}^{1/p}$ for any $1 \leq p < \infty$. Its \emph{$\ell_\infty$-norm} is defined as $\vabs{\bx}_\infty := \max_{i \in [m]} |x_i|$.

\subsection{Tools from Differential Privacy}

Here we note some useful facts regarding differential privacy \cite{DBLP:journals/jpc/DworkMNS16,DBLP:conf/eurocrypt/DworkKMMN06,DBLP:books/sp/17/Vadhan17}.

\begin{fact}[Post-Processing]\label{fct:post-processing}
Let $\Acal_1$ be an $(\eps,\delta)$-DP algorithm and $\Acal_2$ be a (randomized) post-processing algorithm. Then the algorithm $\Acal(\bx)=\Acal_2(\Acal_1(\bx))$ is still an $(\eps,\delta)$-DP algorithm.
\end{fact}

\begin{fact}[Group Privacy]\label{fct:group_privacy}
Let $\Acal$ be an $(\eps,\delta)$-DP algorithm and $\bx,\bx'$ be two arbitrary inputs.
Define $k=\vabs{\bx-\bx'}_1$. Then for any measurable subset $S$ of $\Acal$'s range, we have
$$
\Pr\sbra{\Acal(\bx)\in S}\le e^{k\cdot\eps}\cdot\Pr\sbra{\Acal(\bx')\in S}+\delta\cdot\frac{e^{k\cdot\eps}-1}{e^\eps-1}.
$$
\end{fact}

\begin{fact}[Basic Composition]\label{fct:basic_composition}
Let $\Acal_1$ be an $(\eps_1,\delta_1)$-DP algorithm and $\Acal_2$ be an $(\eps_2,\delta_2)$-DP algorithm. Then $\Acal(\bx)=\pbra{\Acal_1(\bx),\Acal_2(\Acal_1(\bx),\bx)}$ is an $(\eps_1+\eps_2,\delta_1+\delta_2)$-DP algorithm.
\end{fact}

\begin{fact}[Parallel Composition]\label{fct:parallel_composition}
Let $\Acal_1$ be an $(\eps_1,\delta_1)$-DP algorithm and $\Acal_2$ be an $(\eps_2,\delta_2)$-DP algorithm. Assume $\Acal_1$ and $\Acal_2$ depend on disjoint subsets of input coordinates. Then the algorithm $\Acal(\bx)=\pbra{\Acal_1(\bx),\Acal_2(\Acal_1(\bx),\bx)}$ is a $(\max\cbra{\eps_1,\eps_2},\max\cbra{\delta_1,\delta_2})$-DP algorithm.
\end{fact}

Two of the most ubiquitous mechanisms in DP are the Laplace and Gaussian mechanisms \cite{DBLP:conf/tamc/Dwork08,DBLP:conf/eurocrypt/DworkKMMN06,DBLP:journals/fttcs/DworkR14}. 
We use $\Lap(\sigma)$ to denote the \emph{Laplace distribution} with parameter $\sigma$, whose density function is $\frac1{2\sigma}\exp\pbra{-\frac{|x|}\sigma}$.
We use $\Normal(\mu,\sigma^2)$ to denote the \emph{Gaussian distribution} with mean $\mu$ and variance $\sigma^2$, whose density function is $\frac1{\sigma\sqrt{2\pi}}\exp\pbra{-\frac12\frac{(x-\mu)^2}{\sigma^2}}$.

For matrix $\bA$ and $p \geq 1$, we  use $\vabs{\bA}_{\linflp}$ to denote the maximum $\ell_p$-norm among all column vectors of $\bA$.  The Laplace and Gaussian mechanisms for linear queries are stated next.

\begin{lemma}[Laplace Mechanism, \cite{DBLP:conf/tamc/Dwork08,DBLP:journals/vldb/LiMHMR15}]\label{lem:laplace_mechanism}
For the $\bW$-linear query problem, the algorithm that outputs $\bW \bx + \bz$ is $\eps$-DP, where each entry of $\bz$ is drawn i.i.d. from $\Lap\pbra{\vabs{\bW}_{\linflone}/\eps}$.
\end{lemma}

\begin{lemma}[Gaussian Mechanism, \cite{DBLP:conf/eurocrypt/DworkKMMN06,DBLP:journals/fttcs/DworkR14}]\label{lem:gaussian_mechanism}
Assume $\eps, \delta \in(0,1)$.
For the $\bW$-linear query problem, the algorithm that outputs $\bW \bx + \bz$ is $(\eps,\delta)$-DP, where each entry of $\bz$ is drawn i.i.d. from $\Normal\pbra{0, 2\ln \pbra{1.25/\delta}\vabs{\bW}_{\linfltwo}^2/\eps^2}$.
\end{lemma}

Recall that for a tree $\cT$, we let $\bW^{\cT} \in \bin^{\nodes(\cT) \times \leaves(\cT)}$ be the indicator matrix whether a leaf is a descendant of (or itself) a node. 
This represents the tree aggregation problem as $\bW^{\cT}$-linear queries. Observe also that 
$\vabs{\bW^{\cT}}_{\linflone} = d$ and $\vabs{\bW^{\cT}}_{\linfltwo} = \sqrt{d}$.
Therefore, we can apply \Cref{lem:laplace_mechanism} and \Cref{lem:gaussian_mechanism} (together with tail bounds for Laplace and Gaussian distributions) to obtain the following baselines for tree aggregation.

\begin{corollary}[Baseline Algorithms]\label{cor:baseline_algorithms}
For the tree aggregation problem, there exists an $\eps$-DP (resp., $(\eps,\delta)$-DP) algorithm with $\mRMSE$ $O(d/\eps)$ (resp., $O(\sqrt{d\cdot\log(1/\delta)}/\eps)$).
\end{corollary}

We will also use the Laplace mechanism with a bounded range.
For any $R>0$, we use $\TruncLap(\sigma,R)$ to denote the \emph{truncated Laplace distribution} with parameter $\sigma$ and range $[-R,R]$, whose density function is proportional to $\exp\pbra{-|x|/\sigma}$ for $x\in[-R,R]$ and is $0$ if $|x|>R$.
Note that $\Lap(\sigma)=\TruncLap(\sigma,+\infty)$.

\begin{lemma}[Truncated Laplace Mechanism, \cite{DBLP:conf/aistats/GengDGK20}]\label{lem:truncated_laplace_mechanism}
The algorithm that, on input $x\in\Zbb$, outputs $x+z$ is $(\eps,\delta)$-DP, where $z\sim\TruncLap\pbra{\frac1\eps,\frac1\eps\ln\pbra{1+\frac{e^\eps-1}{2\delta}}}$.
\end{lemma}

Our algorithm will also make use of the celebrated \emph{sparse vector technique} \cite{DBLP:conf/stoc/DworkNRRV09}. For convenience, we apply it in a black-box way as the following oracle.

\begin{lemma}[Sparse Vector Technique, \cite{DBLP:conf/stoc/DworkNRRV09,DBLP:journals/fttcs/DworkR14}]\label{lem:svt}
There exists an $\eps$-DP algorithm\\
\Sparse{$\bx, \cbra{f_i};\eta,c,\tau,\eps$} such that:
\begin{itemize}
\item \textsc{Input.} A dataset $\bx$, an adaptively chosen stream $\cbra{f_i}_{i=1,\ldots,d}$ of sensitivity-$1$ queries, error probability $\eta > 0$, a cutoff point $c\ge1$, a threshold $\tau$, and a privacy bound $\eps>0$.\footnote{We say a query $f$ is \emph{sensitivity-1} if $\abs{f(\bx)-f(\bx')}\le1$ for any two neighboring inputs $\bx,\bx'$.}
\item \textsc{Output.} A stream  $\cbra{a_i}_{i=1,\ldots,d}\in\binBT^*$ of on-the-fly answers.
\item \textsc{Accuracy.}
Let $i^*$ be the index of the $c$th $\top$ in $\cbra{a_i}_{i\in[d]}$;
if there are less than $c$ $\top$'s, let $i^* = d$.
Then, for $\Delta = 8c/\eps\cdot\ln\pbra{2d/\eta}$, with probability at least $1-\eta$ the following holds for all $i\le i^*$:
If $a_i=\top$, then $f_i(\bx)\ge\tau-\Delta$;
otherwise (i.e., $a_i=\bot$) $f_i(\bx)<\tau+\Delta$.
\end{itemize}
\end{lemma}

The \Sparse{} algorithm is in \cite[Algorithm 2]{DBLP:journals/fttcs/DworkR14} with $\delta=0$, where its privacy is proved in \cite[Theorem 3.25]{DBLP:journals/fttcs/DworkR14}. The accuracy part is also immediate from the algorithm description. For completeness, we give a proof in \Cref{app:svt}.
\section{Threshold-Based  Utility}\label{sec:util_formulation}

As mentioned in \cref{sec:proof_overview}, we consider a probabilistic utility guarantee with a threshold value supplied at each node in the tree.  
Then we relate it to the original $\alpha$-$\mRMSE$ notion. 

\begin{problem}[Tree Aggregation with Thresholds]\label{prob:estimation}
Consider the tree aggregation problem (i.e., \Cref{prob:tree_counting_problem}), wherein additionally, we have a threshold $\tau_u\ge0$ corresponding to each node $u\in\nodes(\Tcal)$ (both internal nodes and leaves).
The  desired output for the problem is an estimate
$\tbw \in \Rbb^{\nodes(\Tcal)}$ of $\bw$ as before.

For parameters $\eta, \alpha \in [0, 1)$, we say that an algorithm for tree aggregation is \emph{$(\alpha, \eta)$-accurate} (w.r.t. the given thresholds) if its output vector $\tbw \in \R^{\nodes(\cT)}$ satisfies the following:
\vspace{-15pt}
\begin{equation}
\label{eq:thresh}
\text{For all } u \in \nodes(\Tcal) ~~:~~ \Pr\bigg[\abs{\tilde w_u-w_u}\le\alpha\cdot\max\cbra{w_u,\tau_u} \bigg] ~\geq~ 1 - \eta.
\end{equation}
\end{problem}
\vspace{-5pt}

Unless otherwise specified, for the rest of the paper, we assume that the input to the tree aggregation problem includes a threshold at each node in the tree.  Let  $\taumin := \min_{u\in\nodes(\Tcal)}\tau_u$ and $\taumax:=\max_{u\in\nodes(\Tcal)} \tau_u$.  Now, we show that any algorithm with low $\alpha$-$\mRMSE$ also yields a certain utility guarantee in the sense of \Cref{eq:thresh}.

\begin{lemma}[Proof in \Cref{apx:proof_from_mrmse}]\label{lem:from_mrmse}
For any $\alpha' > \alpha \geq 0$, $\eta > 0$, any tree aggregation algorithm with $\alpha$-$\mRMSE$ at most $(\alpha' - \alpha)\sqrt\eta\cdot \taumin$ is also $(\alpha', \eta)$-accurate.
\end{lemma}

In contrast to the above bound, our algorithms will satisfy much stronger exponential tail bounds, which will be clear in the next section.
To complement \Cref{lem:from_mrmse}, we show that any algorithm that is $(\alpha,\eta)$-accurate, as per \Cref{eq:thresh}, can be made having small $\alpha$-$\mRMSE$.

\begin{lemma}[Proof in \Cref{apx:proof_to_mrmse}]\label{lem:to_mrmse}
If there is an $(\eps/2, \delta/2)$-DP algorithm that is $(\alpha, \eta)$-accurate for tree aggregation, then there is an $(\eps, \delta)$-DP algorithm for tree aggregation with $\alpha$-$\mRMSE$ at most $O\pbra{\alpha\cdot\taumax + d\sqrt\eta\cdot\pbra{1+\frac1\eps\log\pbra{\frac1\delta}}}$.
\end{lemma}

With the above two lemmas in mind, it essentially suffices for us to consider the accuracy notion in \Cref{eq:thresh}, which will be convenient for the rest of the paper.

\section{Upper Bounds}\label{sec:upper_bounds}

In this section, we present a new algorithm for tree aggregation in the approximate-DP setting, achieving a significant improvement over the baseline algorithm (i.e., \Cref{cor:baseline_algorithms}).

\begin{theorem}\label{thm:upper_bound}
Let $\alpha>0$ be a parameter.
For any $\eps>0$ and $\delta\in(0,1)$, there is an $(\eps,\delta)$-DP algorithm for tree aggregation with $\alpha$-$\mRMSE$ at most $O\pbra{\frac1{\alpha^3}\cdot\pbra{\frac1\eps\log\pbra{\frac{2d}\delta}+1}}$.
\end{theorem}

By \Cref{lem:to_mrmse}, it suffices to design an efficient algorithm for \Cref{prob:estimation} with small thresholds for every node in the tree.
This will be the focus of the section.
To this end, we will reduce the estimation problem to a classification problem, and design algorithms for the classification task.
We first present the classification algorithm in \cref{sec:a_classification_problem}, then describe the reduction in \cref{sec:a_reduction}, and finally put them together in \cref{sec:putting_everything_together}.

\subsection{A Classification Problem}\label{sec:a_classification_problem}

For later reduction, the classification task here needs to have a stronger notion of success: the classification for nodes should be correct \emph{simultaneously} with high probability, 
whereas \Cref{prob:estimation} only requires the estimation of any fixed node to be correct with high probability.

\begin{problem}\label{prob:classification}
Let $\Tcal$ be a tree of depth $d$ and arity $k$. Let $\tau\ge0$ be the (common) threshold for all nodes. Let $M>0,\eta\in[0,1/2),\alpha\in[0,1)$ be parameters.

The input to the problem is non-negative integer values $x_v$ for each $v\in\leaves(\Tcal)$. For each node $u$, its weight $w_u$ is $w_u=\sum_{v\text{\rm{ is a leaf under }}u}x_v$.

The desired output is a vector $\bw'\in\binBT^{\nodes(\Tcal)}$ that with probability at least $1-\eta$ satisfies the following: 
when the weight of the root of $\Tcal$ is at most $M$, for each $u\in\nodes(\Tcal)$,
\begin{itemize}
\item if $w_u\ge(1+\alpha)\cdot\tau$, then $w'_u=\top$;
\item if $w_u<(1-\alpha)\cdot\tau$, then $w'_u=\bot$;
\item otherwise (i.e., $(1-\alpha)\cdot\tau\le w_u<(1+\alpha)\cdot\tau$), $w'_u$ can be arbitrary.
\end{itemize}
\end{problem}

We now present our algorithm for this classification problem and its guarantees.

\begin{lemma}\label{lem:classification}
There is an $(\eps,\delta)$-DP algorithm \Classification{$\Tcal;M,\eta,\alpha,\tau,\eps,\delta$} such that it solves \Cref{prob:classification} assuming
$
\tau~\ge~\sqrt{\frac{2M}{\alpha\eps}} \cdot \max\cbra{
	\sqrt{48 \ln\pbra{\frac{2d}\eta}},\ 
	\sqrt{6 \ln\pbra{1+\frac{e^{\eps/2}-1}\delta}}
}.
$
\end{lemma}
\begin{proof}
Without loss of generality we assume $\alpha\le1/2$.
Note that if $M<\tau$, then we can simply set $w_u'\gets\bot$ for all $u\in\nodes(\Tcal)$.
Therefore we assume without loss of generality $M\ge\tau$ from now on.
For any node $u$, let $\mathsf{depth}(u)$ denote the number of nodes on the path from the root to $u$.
\Cref{alg:classification} contains the formal
description. 

\begin{algorithm}[ht]
\caption{\textsf{Classification} \label{alg:classification}}
\DontPrintSemicolon
\LinesNumbered
\KwIn{$\Tcal$ and parameters $M,\eta,\alpha,\tau,\eps,\delta$ described in \Cref{prob:classification}}
\KwOut{$w_u'\in\binBT$ for all $u\in\nodes(\Tcal)$}\vspace{3mm}
\lIf{$M<\tau$}{set $w_u' \gets \bot$ for all $u\in\nodes(\Tcal)$ and \Return $\bw'$}
Set $c\gets\frac M{(1-\alpha)\tau}$ and
\vspace{-10pt}
$$
\Delta\gets\tfrac{16c}\eps\ln\pbra{\tfrac{2d}\eta},
\quad
R\gets\tfrac{2c}\eps\ln\pbra{1+\tfrac{c\cdot(e^{\eps/(2c)}-1)}\delta}
$$
\vspace{-20pt}\;
Define $\Scal_i\gets\cbra{u\in\nodes(\Tcal) \,\mid\, \mathsf{depth}(u) = i}$ for each $i\in[d]$\;
\ForEach{$i=d$ \KwTo $1$}{
\lIf{$\Scal_i = \emptyset$}{\textbf{continue}}
Define query $f_i\gets\max_{u\in\Scal_i}w_u$\;
Get $a_i\gets\Sparse(\bx,f_i;\eta,c,\tau,\eps/2)$
\tcc*{$\bx$ is the values of $\leaves(\Tcal)$}
\eIf{$a_i=\bot$}{
Set $w'_u\gets\bot$ for all $u\in\Scal_i$ and update $\Scal_i\gets\emptyset$
}(\tcc*[f]{$a_i = \top$}){
\ForEach{$u\in\Scal_i$}{
Compute $\tilde w_u\gets w_u+\TruncLap(2c/\eps,R)$\;
\eIf{$\tilde w_u\ge\tau-\Delta-R$}{
Set $w'_v\gets\top$ and remove $v$ from $\Scal_{i_v}$ for each $u$'s ancestor $v$ (including $u$ itself) where $i_v = \mathsf{depth}(v)$.
}(\tcc*[f]{$\tilde w_u<\tau-\Delta-R$}){
Set $w'_u\gets\bot$ and remove $u$ from $\Scal_i$
}
}
}
}
\Return{$\bw'$}
\end{algorithm}

We first prove the privacy bound.
By \Cref{lem:svt}, \textsf{Line 7} is $\eps/2$-DP.
On the other hand, by \Cref{lem:truncated_laplace_mechanism} and \Cref{fct:parallel_composition}, \textsf{Lines 11--18} are $(\eps/(2c),\delta/c)$-DP; and since they are executed at most $c$ times, they are $(\eps/2,\delta)$-DP in total.
Therefore by \Cref{fct:basic_composition}, \Cref{alg:classification} is $(\eps,\delta)$-DP.

Now we turn to the correctness of \Cref{alg:classification}.
Define $i^*$ to be the index of the $c$th $\top$ in $a_d,a_{d-1},\ldots$.
If there are less than $c$ $\top$'s, let $i^*\ge1$ be the index of the last query.
Define $\Ecal$ to be the event that the following holds for any $i\ge i^*$:\footnote{Note that $i$ goes from $d$ down to $1$ in our algorithm.}
If $a_i=\top$, then $f_i\ge\tau-\Delta$; and if $a_i=\bot$, then $f_i<\tau+\Delta$.
By \Cref{lem:svt}, $\Pr\sbra{\Ecal}\ge1-\eta$.

We first show, conditioned on $\Ecal$, there are always less than $c$ $\top$'s.
Assume towards contradiction that there are $c$ $\top$'s.
Then, conditioned on $\Ecal$, for any $a_i=\top$, there exists some $u\in\Scal_i$ such that $w_u=f_i\ge\tau-\Delta$. 
Therefore $\tilde w_u\ge w_u-R\ge\tau-\Delta-R$, which implies it will be assigned to $\top$ on \textsf{Line 14}.
By design, all these $u$'s satisfying \textsf{Line 13} form a subset of $\Tcal$ where none of them is an ancestor of another. 
Therefore the weight of the root is lower bounded by the total weights of these nodes, which is at least $c\cdot(\tau-\Delta)$ but at most $M$. On the other hand, by the assumption on $M$ and assuming $\Delta<\alpha\cdot\tau$, we also have $c>\frac M{\tau-\Delta}$, which gives a contradiction.

Then for the correctness part, it suffices to show for any fixed node $v\in\nodes(\Tcal)$, we have $w'_v=\top$ if $w_v\ge(1+\alpha)\cdot\tau$, and $w'_v=\bot$ if $w_v<(1-\alpha)\cdot\tau$:
\begin{itemize}
\item \textsc{Case $w_v\ge(1+\alpha)\cdot\tau$.}
Assume towards contradiction that $w'_v=\bot$. 
Let $i_v = \mathsf{depth}(v)$. Then it means when $i=i_v$ on \textsf{Line 4}, $v\in\Scal_i$. Thus $f_i\ge w_v\ge(1+\alpha)\cdot\tau$.
On the other hand since $w'_v=\bot$, we must proceed to \textsf{Line 9}.
Conditioned on $\Ecal$, this implies $f_i<\tau+\Delta$,
which is a contradiction assuming $\Delta\le\alpha\cdot\tau$.
\item \textsc{Case $w_v<(1-\alpha)\cdot\tau$.}
Assume towards contradiction that $w'_v=\top$. Let $r\in\nodes(\Tcal)$ be the deepest node in the subtree below $v$ that is assigned $\top$. Let $i_r = \mathsf{depth}(r)$. Then it means when $i=i_r$ we execute \textsf{Line 14} for $r$.
Thus $w_r+R\ge\tilde w_r\ge\tau-\Delta-R$.
Meanwhile, we also have $w_r\le w_v<(1-\alpha)\cdot\tau$, which is a contradiction assuming $\Delta+2R\le\alpha\cdot\tau$.
\end{itemize}

Thus it suffices to make sure $\Delta,R\le\alpha\cdot\tau/3$.
Since $M\ge\tau$, we have $c\ge1$ and $c\cdot\pbra{e^{\eps/(2c)}-1}\le e^{\eps/2}-1$,
which gives the assumption in the statement by rearranging terms and noticing $1-\alpha\ge1/2$.
\end{proof}

Eventually, we will use \Classification{$\Fcal;M,\eta,\alpha,\tau,\eps,\delta$} algorithm on a forest $\Fcal$ of disjoint trees with the same set of parameters.
There we do not need all nodes in $\Fcal$ to be classified correctly. 
Instead, it suffices to have all nodes in any $\Tcal\in\Fcal$ classified correctly.

\begin{problem}\label{prob:classification_forest}
Let $\Fcal=\cbra{\Tcal_1,\Tcal_2,\ldots}$ be a forest of disjoint trees of depth $d$ and arity $k$. Let $\tau\ge0$ be the threshold for every node. Let $M>0,\eta\in[0,1/2),\alpha\in[0,1)$ be parameters.

The input to the problem is non-negative integer values $x_v$ for each leaf $v$ in $\Fcal$. For each node $u$, its weight $w_u$ is $w_u=\sum_{v\text{\rm{ is a leaf under }}u}x_v$.

The desired output is a vector $\bw'\in\binBT^{\nodes(\Fcal)}$ that, for any $\Tcal\in\Fcal$ with probability at least $1-\eta$, satisfies the following: 
when the root of $\Tcal$ has weight at most $M$, for each $u\in\nodes(\Tcal)$,
\begin{itemize}
\item if $w_u\ge(1+\alpha)\cdot\tau$, then $w'_u=\top$;
\item if $w_u<(1-\alpha)\cdot\tau$, then $w'_u=\bot$;
\item otherwise (i.e., $(1-\alpha)\cdot\tau\le w_u<(1+\alpha)\cdot\tau$), $w'_u$ can be arbitrary.
\end{itemize}
\end{problem}

The algorithm for \Cref{prob:classification_forest} is simply running \Classification{$\Tcal;M,\eta,\alpha,\tau,\eps,\delta$} for each $\Tcal\in\Fcal$.
Since the trees are disjoint, the privacy bound follows from \Cref{fct:parallel_composition}.
Therefore we omit the proof and summarize the following.

\begin{corollary}\label{cor:classification}
There is an $(\eps,\delta)$-DP algorithm \Classification{$\Fcal;M,\eta,\alpha,\tau,\eps,\delta$} such that it solves \Cref{prob:classification_forest} assuming 
$\tau\ge\sqrt{\frac{2M}{\alpha\eps}} \cdot \max\cbra{
	\sqrt{48\ln\pbra{\frac{2d}\eta}},
	\sqrt{6\ln\pbra{1+\frac{e^{\eps/2}-1}\delta}}
}$.
\end{corollary}

\subsection{A Reduction from Estimation to Classification}\label{sec:a_reduction}

Now we present the reduction algorithm from the estimation problem (i.e., \Cref{prob:estimation}) to the classification problem (i.e., \Cref{prob:classification_forest}).
The reduction here is given with large flexibility for choosing parameters.
Later we will design geometric convergent sequences for simplicity of calculation and derive the final bounds.

\begin{lemma}\label{lem:reduction}
Let $\ell\ge1$ be an integer.
Let $M,M_0$, and $(M_i,\eta_i,\alpha_i,\tau_i,\eps_i,\delta_i)_{i\in[\ell]}$ be a sequence of parameters.
There is an $(\eps, \delta)$-DP algorithm \Reduction{$\Tcal;\ell,\pbra{M_i,\eta_i,\alpha_i,\tau_i,\eps_i,\delta_i}_{i\in[\ell]}$}, where $(\eps,\delta)=\pbra{\sum_{i=1}^\ell\eps_i,\sum_{i=1}^\ell\delta_i}$ such that it solves \Cref{prob:estimation} by carefully combining results from \Classification{$\cdot;M_i,\eta_i,\alpha_i,\tau_i,\eps_i,\delta_i$}'s and assuming the weight of the root of $\Tcal$ is at most $M$ and
\begin{align}
&\tau_i~\ge~ \sqrt{\tfrac{2M_i}{\alpha_i\eps_i}} \cdot \max\cbra{\sqrt{48\ln\pbra{\tfrac{2d}{\eta_i}}},\sqrt{6\ln\pbra{1+\tfrac{e^{\eps_i/2}-1}{\delta_i}}}}
\quad\forall i\in[\ell],
\label{eq:reduction_1}\\
&\eta~\ge~\sum_{i=1}^\ell\eta_i,
\label{eq:reduction_2}\\
&M_i~\ge~(1+\alpha_{i+1})\cdot\tau_{i+1}
\quad\forall i=0,1,\ldots,\ell-1
\quad\text{and}\quad
M_\ell~\ge~ M,
\label{eq:reduction_4}\\
&(1-\alpha_i)\cdot\tau_i~\le~ M_i~\le~(1+\alpha)(1-\alpha_i)\cdot\tau_i
\quad\forall i\in[\ell],
\label{eq:reduction_5}\\
&0~\le~ M_0~\le~\alpha\cdot\taumin.
\label{eq:reduction_6}
\end{align}
\end{lemma}

The reduction algorithm is formalized in \Cref{alg:reduction} and analyzed in \Cref{app:proof_of_lem:reduction}.

\begin{algorithm}[ht]
\caption{ \textsf{Reduction}}\label{alg:reduction}
\DontPrintSemicolon
\LinesNumbered
\KwIn{$\Tcal$ and parameters $\ell,M_0,\pbra{M_i,\eta_i,\alpha_i,\tau_i,\eps_i,\delta_i}_{i\in[\ell]}$ described above}
\KwOut{$\tilde w_u\in\Rbb$ for all node $u\in\nodes(\Tcal)$}
Initialize $\Fcal_\ell\gets\cbra{\Tcal}$\;
\ForEach{$i=\ell$ \KwTo $1$}{
Initialize $\Fcal_{i-1}\gets\emptyset$\;
Compute $\bw'\gets\Classification{$\Fcal_i;M_i,\eta_i,\alpha_i,\tau_i,\eps_i,\delta_i$}$\;
For each node $u$ in $\Fcal_i$, let $\Tcal_u$ be the subtree of $u$ in $\Fcal_i$\;
\ForEach{node $u$ satisfying $w'_u=\top$ and $w'_v=\bot$ for all $v\in\nodes(\Tcal_u)\setminus\cbra{u}$}{
Set $\tilde w_v\gets M_i$ for each $u$'s ancestor $v$ (including $u$ itself) in $\Tcal$\;
Update $\Fcal_{i-1}\gets\Fcal_{i-1}\cup\cbra{\Tcal_1,\Tcal_2,\ldots}$ where $\Tcal_1,\Tcal_2,\ldots$ are the disjoint trees of $\Tcal_u\setminus\cbra{u}$
}
}
Set $\tilde w_v\gets M_0$ for each node $v$ in $\Fcal_0$
\end{algorithm}

\subsection{Putting Everything Together}\label{sec:putting_everything_together}

Now we give the algorithm for \Cref{prob:estimation}.
To this end, we carefully choose parameters and apply \Cref{lem:reduction}, where the required upper bound $M$ is privately estimated with \Cref{lem:truncated_laplace_mechanism}.

\begin{corollary}[Proof in \Cref{app:proof_of_cor:estimation}]\label{cor:estimation}
There is an $(\eps,\delta)$-DP algorithm\\
\Estimation{$\Tcal,\alpha,\eps,\delta,\eta$} such that it solves \Cref{prob:estimation} assuming
$$
\taumin\ge\frac{324\cdot(1+\alpha)^2}{\alpha^4\cdot\eps}\cdot\max\cbra{8\ln\pbra{\frac{4d}{\eta}},\ln\pbra{1+\frac{2\cdot\pbra{e^{\eps/4}-1}}{\delta}}}.
$$
\end{corollary} 

Now we complete the proof of \Cref{thm:upper_bound} using \Cref{lem:to_mrmse} and \Cref{cor:estimation}.

\begin{proof}[Proof of \Cref{thm:upper_bound}]
We first note that if $\alpha\ge1$ then the $\alpha$-$\mRMSE$ is trivially zero by outputting the all-zeros vector. Therefore we assume without loss of generality $\alpha\in(0,1)$.

Let $C>0$ be a constant to be optimized later.
Fix $\eta=d^{-2}$ and 
$\tau=\frac C{\alpha^4}\cdot\pbra{\frac1\eps\log\pbra{\frac{2d}\delta}+1}$.
Let $\eps'=\eps/2$ and $\delta'=\delta/2$.
Since $\alpha\in(0,1)$, $d\ge1$, and $\delta\in(0,1]$, we have
\begin{align*}
\tau^*&:=\tfrac{324\cdot(1+\alpha)^2}{\alpha^4\cdot\eps'}\cdot\max\cbra{8\ln\pbra{\tfrac{4d}{\eta}},\ln\pbra{1+\tfrac{2\cdot\pbra{e^{\eps'/4}-1}}{\delta'}}}
\le O\pbra{\tfrac1{\alpha^4}\cdot\pbra{\tfrac1\eps\log\pbra{\tfrac{2d}\delta}+1}}.
\end{align*}
We set $C$ large enough such that $\tau\ge\tau^*$.
By \Cref{cor:estimation}, there is an $(\eps',\delta')=(\eps/2,\delta/2)$-DP algorithm for \Cref{prob:estimation} when $\tau_u\equiv\tau$ for all $u\in\nodes(\Tcal)$.

The desired $\alpha$-$\mRMSE$ bound now follows from \Cref{lem:to_mrmse} and the parameters above.
\end{proof}
\section{Lower Bounds}\label{sec:lower_bounds}

In this section we prove lower bounds for DP tree aggregation algorithms.
In particular, \Cref{thm:lower_bound} proves pure-DP lower bounds for all $\alpha$-$\mRMSE$ whenever $\alpha\in[0,1)$; and \Cref{thm:add-err-lb} proves approximate-DP lower bounds for additive-only error, i.e., $(\alpha=0)$-$\mRMSE$.

\begin{theorem}[Pure-DP Lower Bound]\label{thm:lower_bound}
Let $\alpha\in[0,1)$ be a parameter.
For any $\eps>0$, any $\eps$-DP algorithm for tree aggregation on the complete depth-$d$ binary tree must incur $\alpha$-$\mRMSE$ at least $\Omega\pbra{(1-\alpha)^2\cdot d/\eps}$.
\end{theorem}

\begin{theorem}[Approximate-DP Lower Bound for $\alpha = 0$]\label{thm:add-err-lb}
For any $\eps > 0$ and any $\delta > 0$ sufficiently small depending on $\eps$, there is a constant $C_{\eps, \delta} > 0$ that any $(\eps, \delta)$-DP algorithm for tree aggregation on the complete depth-$d$ binary tree must incur $\mRMSE$ at least $C_{\eps, \delta} \cdot \sqrt{d}$.
\end{theorem}

\Cref{thm:lower_bound} is proved in \cref{sec:pure_lb}.
The proof of \Cref{thm:add-err-lb} relies on results from \cite{EdmondsNU20} and the factorization norm of the binary tree matrix, which we defer to \Cref{app:apx_lb}.

\subsection{Pure-DP Lower Bound}\label{sec:pure_lb}

To prove \Cref{thm:lower_bound}, by \Cref{lem:from_mrmse} it suffices to rule out DP algorithms for \Cref{prob:estimation} with small thresholds.
Since \Cref{prob:estimation} is interesting on its own, we will present its lower bound in the approximate-DP setting for full generality.

\begin{lemma}\label{lem:estimation_lower_bound_rwalk}
Assume $\Tcal$ in \Cref{prob:estimation} is a complete binary tree of depth $d$. 
Let $D=2\cdot\ceilbra{\taumax/(1-\alpha)}$.
If $\Acal$ is an $(\eps,\delta)$-DP algorithm for \Cref{prob:estimation} and suppose $\eta\le1/8$ and
\begin{equation}\label{eq:estimation_lower_bound_rwalk_assumption}
\delta\cdot\tfrac{e^{\eps\cdot D}-1}{e^\eps-1}\le\tfrac18\cdot2^{-(d-1)\cdot\Hcal(4\eta)},
\end{equation}
then
$$
\taumax=\Omega\pbra{(1-\alpha)\cdot\pbra{d-3-(d-1)\cdot\Hcal(4\eta)}/\eps},
$$
where $\Hcal(x)=x\log\pbra{1/x}+(1-x)\log\pbra{1/(1-x)}$ is the binary entropy function.
\end{lemma}
\begin{proof}
We define input datasets $\bx^1,\ldots,\bx^{2^d}$ where $\bx^i$ assigns $D/2$ to the $i$th leaf and $0$ to the remaining leaves. Let $\Pcal_i$ be the path from root to the $i$th leaf.  We define a randomized decoding algorithm $\Dec$ as follows:
\begin{itemize}
\item $\Dec$ takes the output $\tbw$ of $\Acal$ as input and starts from the root of $\Tcal$.
\item Assume $\Dec$ is at node $u\in\nodes(\Tcal)$.
\begin{itemize}
\item If $u$ is a leaf, then output the index of $u$ among all the leaves.
\item Otherwise let $u_0,u_1$ be the children of $u$ and we divide into the following cases.
\begin{itemize}
\item If $\tilde w_{u_0}\ge\taumax$ and $\tilde w_{u_1}\ge\taumax$, then we move to $u_0$ or $u_1$ with equal probability.
\item If $\tilde w_{u_0}<\taumax$ and $\tilde w_{u_1}<\taumax$, then we move to $u_0$ or $u_1$ with equal probability.
\item Otherwise let $p\in\bin$ be such that $\tilde w_{u_p}\ge\taumax$ and $\tilde w_{u_{1-p}}<\taumax$, then we move to $u_p$ with probability $\kappa$ and to $u_{1-p}$ with probability $1-\kappa$, where
$\kappa = 1 - 4\eta \in[1/2,1]$.
\end{itemize}
\end{itemize}
\end{itemize}

Now we fix an index $i\in[2^d]$.
Let $\Pcal_i$ be $u_1,\ldots,u_d$.
Then for each $j\in[d-1]$, let $u_j^0,u_j^1$ be the children of $u_j$ and assume without loss of generality $u_j^0=u_{j+1}$; then we define the following indicators:
\begin{itemize}
\item $a_j = \ind[ (\tilde w_{u_j^0}\ge\taumax \mbox{ and } \tilde w_{u_j^1}\ge\taumax)$ or $(\tilde w_{u_j^0}<\taumax \mbox{ and } \tilde w_{u_j^1}<\taumax) ]$.
\item $b_j = \ind[ \tilde w_{u_j^0}\ge\taumax \mbox{ and } \tilde w_{u_j^1}<\taumax ]$.
\item $c_j = \ind[ \tilde w_{u_j^0}<\taumax \mbox{ and } \tilde w_{u_j^1}\ge\taumax ]$.
\end{itemize}
Let $A=\sum_ja_j$, $B=\sum_jb_j$, and $C=\sum_jc_j$.
Then it is easy to see $A+B+C=d-1$ and
\begin{equation}\label{eq:estimation_lower_bound_rwalk_1}
\Pr\sbra{\Dec(\Acal(\bx^i))=i}
=\E\sbra{2^{-A}\kappa^B(1-\kappa)^C}
=\kappa^{d-1}\E\sbra{1/(2\kappa)^{d-1-B}(2-2\kappa)^C}.
\end{equation}

Now we further fix the input to be $\bx^i$ defined above.
Then $w_u=D/2$ for $u\in\Pcal_i$ and $w_u=0$ if otherwise.
For $p\in\bin$, consider the event $\Ecal_j^p$: $\abs{\tilde w_{u_j^p}-w_{u_j^p}}\le\alpha\cdot\max\cbra{w_{u_j^p},\tau_{u_j^p}}$.
Hence when $\Ecal_j^0$ happens, we have
$$
\tilde w_{u_j^0}
\ge w_{u_j^0}-\alpha\cdot\max\cbra{w_{u_j^0},\tau_{u_j^0}}
\ge D/2-\alpha\cdot\max\cbra{D/2,\taumax}
=(1-\alpha)\cdot D/2\ge\taumax.
$$
Similarly when $\Ecal_j^1$ happens, we have
$\tilde w_{u_j^1}
\le w_{u_j^1}+\alpha\cdot\max\cbra{w_{u_j^1},\tau_{u_j^1}}
\le\alpha\cdot\taumax<\taumax$.
Meanwhile by the definition of \Cref{prob:estimation}, we know $\Pr\sbra{\Ecal_j^p}\ge1-\eta$. Thus
$$
\Pr\sbra{b_j=1}\ge\Pr\sbra{\Ecal_j^0 \land\Ecal_j^1}\ge1-2\eta,
\quad\text{and}\quad
\Pr\sbra{c_j=1}\le\Pr\sbra{\neg\Ecal_j^0\land\neg\Ecal_j^1}\le\Pr\sbra{\neg\Ecal_j^0}\le\eta,
$$
which implies $\E\sbra{d-1-B}\le2\eta\cdot(d-1)$ and $\E\sbra{C}\le\eta\cdot(d-1)$.
Note that $d-1-B\ge0$. 
Define the event $\Ecal$: $d-1-B\le4\eta\cdot(d-1)$ and $C\le4\eta\cdot(d-1)$.
Then by Markov's inequality and a union bound, we have $\Pr\sbra{\Ecal}\ge1-1/2-1/4=1/4$.
Plugging into \Cref{eq:estimation_lower_bound_rwalk_1}, we have
\begin{align*}
\Pr\sbra{\Dec(\Acal(\bx^i))=i}
&\ge\kappa^{d-1}/4\cdot\E\sbra{1/(2\kappa)^{d-1-B}(2-2\kappa)^C ~\mid~ \Ecal}\\
&\ge\kappa^{d-1}/4\cdot1/(2\kappa)^{4\eta\cdot(d-1)}\cdot(2-2\kappa)^{4\eta\cdot(d-1)}
\tag{since $\kappa\in[1/2,1]$}\\
&=\frac14\cdot2^{-(d-1)\cdot\Hcal(4\eta)}.
\tag{setting $\kappa=1-4\eta\in[1/2,1]$}
\end{align*}

Since $\sum_{i'}\Pr\sbra{\Dec(\Acal(\bx^i))=i'}=1$, by an averaging argument there exists an $i^*$ such that $\Pr\sbra{\Dec(\Acal(\bx^i))=i^*}\le2^{-d}$.
Since $\vabs{\bx^i-\bx^{i^*}}_1\in\cbra{0,D}$, by \Cref{fct:group_privacy} we have 
$$
\Pr\sbra{\Dec(\Acal(\bx^{i^*}))=i^*}
\le
\Pr\sbra{\Dec(\Acal(\bx^i))=i^*}\cdot e^{\eps\cdot D}+\delta\cdot\tfrac{e^{\eps\cdot D}-1}{e^\eps-1}
\le
2^{-d}\cdot e^{\eps\cdot D}+\delta\cdot\tfrac{e^{\eps\cdot D}-1}{e^\eps-1}.
$$
In all, we have
\begin{align*}
\tfrac14\cdot2^{-(d-1)\cdot\Hcal(4\eta)}
&\le\Pr\sbra{\Dec(\Acal(\bx^{i^*}))=i^*}
\le2^{-d}\cdot e^{\eps\cdot D}+\delta\cdot\tfrac{e^{\eps\cdot D}-1}{e^\eps-1},
\end{align*}
which proves the bound after plugging in Assumption \Cref{eq:estimation_lower_bound_rwalk_assumption} and rearranging the terms.
\end{proof}

To deal with general $\eta$ from \Cref{prob:estimation}, we simply run independent copies of the algorithm to decrease the error probability.

\begin{corollary}[Proof in \Cref{app:cor:estimation_lower_bound_rwalk}]\label{cor:estimation_lower_bound_rwalk}
Assume $\Tcal$ in \Cref{prob:estimation} is a complete binary tree of depth $d$. 
Define $D=2\cdot\ceilbra{\frac{\taumax}{1-\alpha}}$ and $s=\ceilbra{\frac{\ln(4/\kappa)}{2\cdot(1/2-\eta)^2}}$ for any parameter $\kappa\in(0,1/2]$.
If $\Acal$ is an $(\eps,\delta)$-DP algorithm for \Cref{prob:estimation} and suppose 
$s\cdot\delta\cdot\frac{e^{s\cdot\eps\cdot D}-1}{e^{s\cdot\eps}-1}\le\frac18\cdot2^{-(d-1)\cdot\Hcal(\kappa)}$,
then
$$
\taumax=\Omega\pbra{\frac{(1-\alpha)\cdot(1/2-\eta)^2\cdot(d-3-(d-1)\Hcal(\kappa))}{\eps\cdot\ln(4/\kappa)}}.
$$
\end{corollary}

\begin{remark}[General Tree Structures]
The proof above works almost identically for binary tree $\Tcal$ that is not necessarily complete, where the bound is simply replacing $d-3$ with $\log(|\leaves(\Tcal)|/8)$.
To deal with general arity $k$, we can embed a binary tree $\Tcal'$ into $\Tcal$ where we say $\Tcal$ \emph{embeds} a tree $\Tcal'$ if we can obtain $\Tcal'$ from $\Tcal$ by deleting nodes and edges.
Then we can ignore nodes in $\nodes(\Tcal)\setminus\nodes(\Tcal')$ and obtain lower bounds for $\Tcal'$.
\end{remark}

Now we are ready to establish \Cref{thm:lower_bound} using \Cref{lem:from_mrmse} and \Cref{cor:estimation_lower_bound_rwalk}.

\begin{proof}[Proof of \Cref{thm:lower_bound}]
Let $\Tcal$ be the complete binary tree of depth $d$.
Let $C,\tau,\eta$ be parameters to be optimized later.
We consider \Cref{prob:estimation} where $\tau_u\equiv\tau$ for all $u\in\nodes(\Tcal)$.

Assume towards contradiction that there is an $\eps$-DP tree aggregation algorithm with $\alpha$-$\mRMSE$ at most $C\cdot(1-\alpha)^2d/\eps$.
Then by \Cref{lem:from_mrmse}, the algorithm is also $(\alpha',\eta)$-accurate for \Cref{prob:estimation} if $\alpha'>\alpha$ and $(\alpha'-\alpha)\sqrt\eta\cdot\tau\le C\cdot(1-\alpha)^2d/\eps$.

Now we set $\eta=1/4$ and apply \Cref{cor:estimation_lower_bound_rwalk} with $\delta=0,\kappa=1/4$. 
This gives a lower bound $\tau=\Omega\pbra{(1-\alpha')\cdot d/\eps}$,
which means $\frac{C\cdot(1-\alpha)^2\cdot d}\eps\ge\Omega\pbra{\frac{(\alpha'-\alpha)(1-\alpha')\cdot d}\eps}$.
Then we set $\alpha'=(1+\alpha)/2>\alpha$ and $C=O(1)$ small enough to derive a contradiction.
\end{proof}

\begin{remark}[$\log(1/\delta)$ Factor in the Approximate-DP Algorithm]\label{rmk:optimality_apx}
As mentioned in \cref{sec:our_results}, our improved $(\eps,\delta)$-DP algorithm (See \Cref{thm:upper_bound_informal}) has a $\log(1/\delta)$ factor, which is worse than the $\sqrt{\log(1/\delta)}$ factor in the Gaussian mechanism (see \Cref{cor:baseline_algorithms}). One may wonder if we can further improve the dependency on $\delta$ to, say, $\sqrt{\log(1/\delta)}$, without influencing the other parameters.
Combining \Cref{cor:estimation_lower_bound_rwalk}, we show this is in some sense impossible.

Consider the complete binary tree of depth $d$.
Let $\delta=2^{-\Omega(d)}$.
We consider the case where $\alpha,\eta$ are constants and all the $\tau_u$'s are equal to $\tau$.
In the approximate-DP setting, we naturally seek bounds better than the ones in the pure-DP setting (recall we can obtain $\tau=O(d/\eps)$ from \Cref{cor:baseline_algorithms}). Thus we assume $\tau=O(d/\eps)$ in advance.

Then for a suitable choice of $\kappa=\Theta(1)$, the condition in \Cref{cor:estimation_lower_bound_rwalk} holds, which gives a lower bound $\tau=\Omega(d/\eps)$ for \Cref{prob:estimation}.
Then by \Cref{lem:from_mrmse}, this rules out the possibility of improving the dependency on $\delta$ in \Cref{thm:upper_bound_informal} without worsening the dependency on $d$.
\end{remark}
\section{Conclusions}\label{sec:conclusion}

We study the problem of privately estimating counts in hierarchical data, and give several algorithms and lower bounds.
We propose a new error measure that takes the multiplicative error into account.
The commonly used $\ell_2^2$-error measure in evaluating utilities of DP mechanisms allows some queries to have huge error.
On the other hand, the standard measure $\ell_\infty$-error has a ``union bound issue'' on particularly long output vector (which is the case in Census and Ads applications).

To mitigate these weaknesses, we propose $\alpha$-multiplicative root mean squared error ($\alpha$-$\mRMSE$).
Then we examine the standard Laplace mechanism for pure-DP and Gaussian mechanism for approximate-DP, and prove their optimality.
Informally, we show Laplace mechanism already achieves optimal bounds in the pure-DP setting for all multiplicative factor $\alpha$ and Gaussian mechanism is optimal in the approximate-DP setting when $\alpha=0$ (i.e., additive-only error).

For the remaining case where we allow $\alpha>0$ and an approximate-DP algorithm, we design a new algorithm with exponential improvements over Gaussian mechanism.
More precisely, Gaussian mechanism incurs $\alpha$-$\mRMSE$ of $O_{\alpha,\eps,\delta}(\sqrt d)$ while our algorithm gives improved bounds of $O_{\alpha,\eps,\delta}(\log(d))$ (and $O\pbra{\frac1{\alpha^3}\cdot\pbra{\frac1\eps\log\pbra{\frac{2d}\delta}+1}}$ specifically).
It remains an interesting question if the dependency on $d$ or $\alpha$ can be improved further.
Indeed, current lower bounds do not preclude bounds of the form $O_{\eps,\delta}\pbra{\frac1\alpha\cdot\log^*(d)}$ or even $O_{\eps,\delta}(1/\alpha)$.

Throughout this work, we assumed that the entries of the input $\bx$ are non-negative.
Another interesting direction is to extend the study to the case where the entries of $\bx$ can be negative. Here, the multiplicative error would be with respect to the absolute value of the true answer. Our algorithms do not apply here and it is unclear whether allowing a multiplicative error can help reduce the additive error in this setting.

\section*{Acknowledgements}
KW wants to thank Xin Lyu for helpful references on the sparse vector technique. We thank anonymous ITCS'23 and ICALP'23 reviewers for helpful feedback.  

\bibliographystyle{alpha}
\bibliography{ref}

\appendix
\section{Proofs}

\subsection[Proof of Lemma 17]{Proof of \Cref{lem:svt}}\label{app:svt}

We explicitly present the algorithm in \Cref{alg:sparse}, which is almost identical to \cite[Algorithm 2]{DBLP:journals/fttcs/DworkR14} with $\delta=0$.
The only change is we do not halt the algorithm after report $c$ $\top$'s.
Instead, we output the default value $\bot$ after $c$ $\top$'s.

\begin{algorithm}[ht]
\caption{\textsf{Sparse}}\label{alg:sparse}
\DontPrintSemicolon
\LinesNumbered
\KwIn{$\bx,\cbra{f_i}$ and parameters $\eta,c,\tau,\eps$ described in \Cref{lem:svt}}
\KwOut{$a_i\in\binBT$ for each query $f_i$}
Initialize $\tilde\tau\gets\tau+\Lap(2c/\eps)$ and $t\gets0$\;
\ForEach{$i=1,2,\ldots$}{
\eIf{$t\ge c$}{
Answer $a_i\gets\bot$
}(\tcc*[f]{$t<c$}){
Set $\tilde f_i\gets f_i(\bx)+\Lap(4c/\eps)$\;
\eIf{$\tilde f_i<\tilde\tau$}{
Answer $a_i\gets\bot$
}(\tcc*[f]{$\tilde f_i\ge\tilde\tau$}){
Answer $a_i\gets\top$\;
Update $\tilde\tau\gets\tau+\Lap(2c/\eps)$ and $t\gets t+1$
}
}
}
\end{algorithm}

\begin{proof}[Proof of \Cref{lem:svt}]
The privacy guarantee is precisely \cite[Theorem 3.25]{DBLP:journals/fttcs/DworkR14}.
Thus we only prove the accuracy bound.

Set
$$
\Delta=\frac{8c}\eps\ln\pbra{\frac{2d}\eta},
$$
with foresight.
Let $\Ecal$ be the event that each noise $\Lap(4c/\eps)$ from \textsf{Line 6} and each noise $\Lap(2c/\eps)$ from \textsf{Lines 1, 11} has absolute value at most $\Delta/2$.
Note that we ignore \textsf{Line 11} when $i=d$. Therefore there are at most $d$ many $\Lap(2c/\eps)$.
By the definition of Laplace distribution and a union bound, we know
$$
\Pr\sbra{\Ecal}\ge1-d\cdot e^{-\frac{\eps\Delta}{8c}}-d\cdot e^{-\frac{\eps\Delta}{4c}}\ge1-\eta.
$$

By the definition of $i^*$, it suffices to consider each $i$ that uses the $\Lap(4c/\eps)$ on \textsf{Line 6}. 
Conditioned on $\Ecal$, we have 
$$
f_i-\frac\Delta2\le\tilde f_i\le f_i+\frac\Delta2
\quad\text{and}\quad
\tau-\frac\Delta2\le\tilde\tau\le\tau+\frac\Delta2.
$$
Thus
\begin{itemize}
\item If $a_i=\top$, then $\tilde f_i\ge\tilde\tau$, which implies $f_i\ge\tau-\Delta$ as desired.
\item If $a_i=\bot$, then $\tilde f_i<\tilde\tau$, which implies $f_i<\tau+\Delta$ as desired.
\qedhere
\end{itemize}
\end{proof}

\subsection[Proof of Lemma 19]{Proof of \Cref{lem:from_mrmse}}\label{apx:proof_from_mrmse}

\begin{proof}
For any $u\in\nodes(\Tcal)$, we have
\begin{align*}
\Pr\sbra{\abs{\tilde w_u-w_u} > \alpha'\cdot\max\cbra{w_u,\tau_u}} 
&\leq \Pr\sbra{\abs{\tilde w_u-w_u} - \alpha\cdot w_u > (\alpha' - \alpha)\cdot\tau_u} \\
&=\Pr\sbra{\max\cbra{\abs{\tilde w_u-w_u} - \alpha\cdot w_u , 0}>(\alpha' - \alpha)\cdot\tau_u}\\
&\leq \frac{\E\sbra{\pbra{\max\cbra{\abs{\tilde w_u-w_u} - \alpha\cdot w_u , 0}}^2}}{(\alpha' - \alpha)^2\cdot\tau_u^2} 
\tag{by Chebyshev's inequality}\\
&\leq \frac{\pbra{\RMSE_{\alpha}(\tw_u, w_u)}^2}{(\alpha' - \alpha)^2\cdot\taumin^2}
\tag{since $\taumin\le\tau_u$}\\
&\leq \eta.
\tag*{\qedhere}
\end{align*}
\end{proof}

\subsection[Proof of Lemma 20]{Proof of \Cref{lem:to_mrmse}}\label{apx:proof_to_mrmse}

\begin{proof}
Let $\eps' = \eps/(2d)$ and $\delta' = \delta/(2d)$.
The algorithm works as follows:
\begin{itemize}
\item Compute $\tw'_u$ for all $u \in \nodes(\cT)$ using the $(\eps/2, \delta/2)$-DP $(\alpha, \eta)$-accurate algorithm for tree aggregation.
\item For each $u \in \nodes(\cT)$, compute $\tw''_u = w_u + z_u$ where 
$$
z_u\sim\TruncLap\pbra{\frac{1}{\eps'}
,R}
\quad\text{and}\quad
R = \frac{1}{\eps'} \ln\pbra{1+\frac{e^{\eps'}-1}{2\delta'}}.
$$
Then output
\begin{align*}
\tw_u =
\begin{cases}
\tw'_u &\text{if } \tw'_u \in \sbra{\tw''_u - R, \tw''_u + R}, \\
\tw''_u - R  &\text{if } \tw'_u < \tw''_u - R, \\
\tw''_u + R &\text{otherwise}.
\end{cases}
\end{align*}
\end{itemize}

To see that the algorithm is $(\eps, \delta)$-DP, recall that $\Tcal$ has depth $d$ and the weights of nodes of the same depth depend on disjoint subsets of the values of the leaves.
Therefore by \Cref{fct:parallel_composition} and \Cref{lem:truncated_laplace_mechanism}, for any fixed depth, the computation of $\tw_u''$ for all $u\in\nodes(\Tcal)$ of the depth is $(\eps',\delta')$-DP.
Then by \Cref{fct:basic_composition}, the computation of $\tw_u''$ for all $u\in\nodes(\Tcal)$ is $(d\eps',d\delta')=(\eps/2,\delta/2)$-DP.
Thus, applying \Cref{fct:basic_composition} again together with \Cref{fct:post-processing}, the entire algorithm is $(\eps, \delta)$-DP.

For the utility guarantee, we fix an arbitrary $u\in\nodes(\Tcal)$ and it suffices to bound its $\alpha$-$\RMSE$.
By the definition of truncated Laplace distribution, we have $w_u \in \sbra{\tw''_u - R, \tw''_u + R}$, which means we always have 
$$
|\tw_u - w_u| \leq |\tw'_u - w_u|
\quad\text{and}\quad
|\tw_u - w_u| \leq 2R.
$$
Thus, we can bound its $\alpha$-$\RMSE$ by
\begin{align*}
\pbra{\RMSE_{\alpha}(\tw_u, w_u)}^2 
&=\E\sbra{\pbra{\max\cbra{\abs{\tw_u - w_u} - \alpha\cdot w_u, 0}}^2} \\
&\leq \Pr\sbra{\abs{\tw_u' - w_u} \leq \alpha\pbra{w_u+\tau_u}} \cdot \pbra{\alpha\tau_u}^2 \\
&\qquad\qquad
+  \Pr\sbra{\abs{\tw_u' - w_u} >\alpha\pbra{w_u+\tau_u}} \cdot 4R^2 
\tag{since $|\tw_u - w_u| \leq |\tw'_u - w_u|$}\\
&\leq \pbra{\alpha\cdot\tau_u}^2 
+  \Pr\sbra{\abs{\tw_u' - w_u} >\alpha\cdot\max\cbra{w_u,\tau_u}} \cdot 4R^2\\
&\leq \pbra{\alpha\cdot\tau_u}^2 + \eta\cdot O\pbra{\frac{d^2}{\eps^2}\cdot\log^2\pbra{1+\frac{2d\cdot\pbra{e^{\eps/(2d)}-1}}\delta}}\\
&\le\pbra{\alpha\cdot\tau_u}^2 + \eta\cdot O\pbra{\frac{d^2}{\eps^2}\cdot\log^2\pbra{1+\frac{e^\eps-1}\delta}}
\tag{since $T\cdot(e^{a/T}-1)\le e^a-1$ for all $a\ge0$ and $T\ge1$}\\
&\le\pbra{\alpha\cdot\tau_u}^2 + \eta\cdot O\pbra{\frac{d^2}{\eps^2}\cdot\log^2\pbra{\frac{e^\eps}\delta}}
\tag{since $\delta\le1$}\\
&\leq (\alpha \cdot \taumax)^2 + O\pbra{d^2\cdot\eta + \frac{d^2\log^2(1/\delta)}{\eps^2}\cdot\eta}.
\tag{since $\tau_u\le\taumax$}
\end{align*}
Thus, we can conclude that $\RMSE_{\alpha}(\tw_u, w_u) \leq O\pbra{\alpha\cdot\taumax + d\sqrt\eta\cdot\pbra{1+\log(1/\delta)/\eps}}$ as desired.
\end{proof}

\subsection[Proof of Lemma 26]{Proof of \Cref{lem:reduction}}\label{app:proof_of_lem:reduction}

\begin{proof}
Recall \Cref{alg:reduction}.
By \Cref{cor:classification}, each \textsf{Line 4} is an $(\eps_i,\delta_i)$-DP algorithm.\footnote{We remark that the privacy guarantee here does not rely on any assumptions on the choice of $M_i,\eta_i,\alpha_i,\tau_i$.}
Thus the privacy bound of \Cref{alg:reduction} follows naturally from \Cref{fct:basic_composition} and \Cref{fct:post-processing}.

Now we turn to the correctness part.
Let $u\in\nodes(\Tcal)$ be a fixed node.
For each $i=0,1,\ldots,\ell$, let $\Tcal^u_i\in\Fcal_i$ be, if exists, the tree containing $u$.
We first prove by induction for $i$ that with probability at least $1-\sum_{j>i}\eta_j$, for any $j\ge i$:
(1) The weight of the root of $\Tcal^u_j$, if well-defined, is at most $M_j$, and
(2) the classification algorithm is correct on $\Tcal^u_{j+1}$, if well-defined.
\begin{itemize}
\item \textsc{Base case $i=\ell$.} 
By Assumption \Cref{eq:reduction_4}, the weight of the root of $\Tcal^u_\ell=\Tcal$ is at most $M\le M_\ell$ as desired.
\item \textsc{Inductive Case $i<\ell$.} 
By the induction hypothesis, with probability at least $1-\sum_{j>i+1}\eta_j$, for any $j\ge i+1$ the weight of the root of $\Tcal^u_j$ is at most $M_j$, and the classification algorithm is correct on $\Tcal_{j+1}^u$.

In particular, the weight of the roots in $\Tcal^u_{i+1}$ is at most $M_{i+1}$.
We also note that for \Classification{$\Fcal_{i+1};M_{i+1},\eta_{i+1},\alpha_{i+1},\tau_{i+1},\eps_{i+1},\delta_{i+1}$},
Assumption \Cref{eq:reduction_1} satisfies the condition of \Cref{cor:classification}.
Thus by \Cref{cor:classification}, with probability at least $1-\eta_{i+1}$ the classification algorithm is correct on $\Tcal^u_{i+1}$.
Hence any node $v\in\nodes(\Tcal^u_{i+1})$ satisfying $w_v\ge(1+\alpha_{i+1})\cdot\tau_{i+1}$ has $w'_v=\top$.
Therefore by the selection method on \textsf{Line 6}, any subtree (including $\Tcal^u_i$) in $\Tcal^u_{i+1}$ put into $\Fcal_i$ has root assigned $\bot$ and has weight at most $(1+\alpha_{i+1})\cdot\tau_{i+1}$, which is at most $M_i$ by Assumption \Cref{eq:reduction_4} as desired.
\end{itemize}

We prove the correctness of our algorithm conditioned on the event that for any $i=0,1,\ldots,\ell$: If $\Tcal^u_i$ is well-defined, then
(1) the weight of the root of $\Tcal^u_i$ is at most $M_i$, and 
(2) the classification algorithm is correct on $\Tcal^u_i$.
By the analysis above and Assumption \Cref{eq:reduction_2}, this happens with probability at least $1-\sum_i\eta_i\ge1-\eta$.
Then we have the following two cases:
\begin{itemize}
\item If we set $\tilde w_u=M_i$ on \textsf{Line 7}, then $(1-\alpha_i)\cdot\tau_i\le w_u\le M_i$ by Assumption \Cref{eq:reduction_5} and
$$
\abs{\tilde w_u-w_u}\le M_i-(1-\alpha_i)\cdot\tau_i\le\alpha\cdot(1-\alpha_i)\cdot\tau_i\le\alpha\cdot w_u\le\alpha\cdot\max\cbra{w_u,\tau_u}.
$$
\item If we set $\tilde w_u=M_0$ on \textsf{Line 11}, then $0\le w_u\le M_0$.
By Assumption \Cref{eq:reduction_6}, this implies
\begin{equation*}
\abs{\tilde w_u-w_u}\le M_0\le\alpha\cdot\taumin\le\alpha\cdot\max\cbra{w_u,\tau_u}.
\tag*{\qedhere}
\end{equation*}
\end{itemize}
\end{proof}

\subsection[Proof of Corollary 27]{Proof of \Cref{cor:estimation}}\label{app:proof_of_cor:estimation}

Let $M$ be a parameter. We first prove \Cref{thm:upper_bound} with upper bound $M$ provided.
\begin{theorem}\label{thm:estimation}
There is an $(\eps,\delta)$-DP algorithm such that it solves \Cref{prob:estimation} assuming the weight of the root of $\Tcal$ is at most $M$ and
$$
\taumin\ge\frac{162\cdot(1+\alpha)^2}{\alpha^4\cdot\eps}\cdot\max\cbra{8\ln\pbra{\frac{4d}{\eta}},\ln\pbra{1+\frac{e^{\eps/2}-1}{\delta}}}.
$$
\end{theorem}
\begin{proof}
We apply \Cref{lem:reduction} with parameters $\ell,M_0$, and $(M_i,\eta_i,\alpha_i,\tau_i,\eps_i,\delta_i)_{i\in[\ell]}$ determined as follows.

For simplicity, we will set $\alpha_i\equiv\beta$ for some parameter $\beta$ to be optimized later.
Then define $M_0=\alpha\cdot\taumin$ and $M_i=(1+\beta)\cdot\tau_{i+1}=(1+\alpha)(1-\beta)\cdot\tau_i$.
Thus
$$
M_i=\alpha\cdot\taumin\cdot\pbra{\frac{(1+\alpha)(1-\beta)}{1+\beta}}^i
\quad\text{and}\quad
\tau_i=\frac{\alpha\cdot\taumin}{1+\beta}\cdot\pbra{\frac{(1+\alpha)(1-\beta)}{1+\beta}}^{i-1}.
$$
Note that we will pick $\beta$ to make sure $(1+\alpha)(1-\beta)>1+\beta$, i.e., $\beta<\frac\alpha{2+\alpha}$.
Then it suffices to set
$$
\ell=\ceilbra{\ln\pbra{\frac M{\alpha\cdot\taumin}}\middle/\ln\pbra{\frac{(1+\alpha)(1-\beta)}{1+\beta}}}.
$$

Now that Assumptions \Cref{eq:reduction_4}, \Cref{eq:reduction_5}, and \Cref{eq:reduction_6} are satisfied, we turn to $\eps_i,\delta_i,\eta_i$.
Rearranging terms in Assumption \Cref{eq:reduction_1}, it is equivalent to
$$
\taumin\ge\tfrac{2(1+\alpha)(1-\beta^2)}\beta\cdot\pbra{\tfrac{1+\beta}{(1+\alpha)(1-\beta)}}^{i-1}\cdot\max\cbra{\frac{48}{\alpha\cdot\eps_i}\ln\pbra{\frac{2d}{\eta_i}},\frac6{\alpha\cdot\eps_i}\ln\pbra{1+\frac{e^{\eps_i/2}-1}{\delta_i}}}.
$$
Then we set
$$
\eta_i=\eta/2^i,
\quad
\eps_i=\frac\eps C\cdot i\cdot\pbra{\frac{(1+\alpha)(1-\beta)}{1+\beta}}^{i-1},
\quad
\delta_i=\frac\delta C\cdot i\cdot\pbra{\frac{(1+\alpha)(1-\beta)}{1+\beta}}^{i-1},
$$
where $C$ is the normalizing factor computed by
$$
C=\sum_{i=1}^{+\infty}i\cdot\pbra{\frac{(1+\alpha)(1-\beta)}{1+\beta}}^{i-1}
=\pbra{1-\frac{(1+\alpha)(1-\beta)}{1+\beta}}^{-2}
=\pbra{\frac{(2+\alpha)\cdot\beta-\alpha}{1+\beta}}^{-2}.
$$
Since $T\cdot\pbra{e^{a/T}-1}\le e^a-1$ holds for all $T\ge1$ and $a\ge0$, we have
$$
\frac{e^{\eps_i/2}-1}{\delta_i}\le\frac{e^{\eps/2}-1}\delta.
$$
Hence it suffices to satisfy
\begin{align*}
\taumin
&\ge\tfrac{2(1+\alpha)(1-\beta^2)}\beta\cdot\pbra{\tfrac{1+\beta}{(1+\alpha)(1-\beta)}}^{i-1}\cdot\max\cbra{\frac{48}{\alpha\cdot\eps_i}\ln\pbra{\frac{2d}{\eta_i}},\frac6{\alpha\cdot\eps_i}\ln\pbra{1+\frac{e^{\eps/2}-1}{\delta}}}\\
&=\tfrac{2(1+\alpha)(1-\beta^2)}\beta\cdot C\cdot\max\cbra{\frac{48}{\alpha\cdot\eps\cdot i}\cdot\pbra{\ln\pbra{\tfrac{2d}{\eta}}+i\ln(2)},\frac6{\alpha\cdot\eps\cdot i}\ln\pbra{1+\tfrac{e^{\eps/2}-1}{\delta}}}.
\end{align*}
The RHS is maximized at $i=1$ which gives
$$
\taumin\ge\frac{2(1+\beta)^2(1-\beta^2)}{\beta\cdot\pbra{(2+\alpha)\cdot\beta-\alpha}^2}\cdot\max\cbra{\frac{48\cdot(1+\alpha)}{\alpha\cdot\eps}\ln\pbra{\frac{4d}{\eta}},\frac{6\cdot(1+\alpha)}{\alpha\cdot\eps}\ln\pbra{1+\frac{e^{\eps/2}-1}{\delta}}}.
$$
Finally we set $\beta$ as $\beta=\frac\alpha{6+5\alpha}$ and notice $2(1+\beta)^2(1-\beta^2)\le2(1+\beta)^3$. 
Then the requirement on $\taumin$ becomes
\begin{equation*}
\taumin\ge\frac{27\cdot(1+\alpha)}{\alpha^3}\cdot\max\cbra{\frac{48\cdot(1+\alpha)}{\alpha\cdot\eps}\ln\pbra{\frac{4d}{\eta}},\frac{6\cdot(1+\alpha)}{\alpha\cdot\eps}\ln\pbra{1+\frac{e^{\eps/2}-1}{\delta}}}.
\tag*{\qedhere}
\end{equation*}
\end{proof}

We remark that the upper bound $M$ is not really necessary as it only logarithmically influences the runtime of our algorithm.
We can discard this assumption by performing a private estimation (e.g., \Cref{lem:truncated_laplace_mechanism}) of the weight of the root of $\Tcal$ in the beginning, which suffices for providing such $M$.
Now we conclude the proof of \Cref{cor:estimation}.

\begin{proof}[Proof of \Cref{cor:estimation}]
The details are given in \Cref{alg:estimation}.

\begin{algorithm}[ht]
\caption{\textsf{Estimation}}\label{alg:estimation}
\DontPrintSemicolon
\LinesNumbered
\KwIn{$\Tcal,\alpha,\eps,\delta,\eta$ described in \Cref{cor:estimation}}
\KwOut{$\tilde w_u\in\Rbb$ for all node $u\in\nodes(\Tcal)$}
Define $R\gets\frac2\eps\ln\pbra{1+\frac{e^{\eps/2}-1}{\delta}}$\;
Compute $M\gets w_\textsf{root}+R+\TruncLap(2/\eps,R)$ where \textsf{root} is the root of $\Tcal$\;
Let 
$$
\beta\gets\frac\alpha{6+5\alpha}
\quad\text{and}\quad
\ell=\ceilbra{\ln\pbra{\frac M{\alpha\cdot\taumin}}\middle/\ln\pbra{\frac{(1+\alpha)(1-\beta)}{1+\beta}}}
$$\;
Set
$$
M_i\gets\alpha\cdot\taumin\cdot\pbra{\frac{(1+\alpha)(1-\beta)}{1+\beta}}^i,
\quad\quad
\tau_i\gets\frac{\alpha\cdot\taumin}{1+\beta}\cdot\pbra{\frac{(1+\alpha)(1-\beta)}{1+\beta}}^i
$$
then set $C\gets\pbra{\frac{(2+\alpha)\cdot\beta-\alpha}{1+\beta}}^{-2}$ and
$$
\eta_i\gets\eta/2^i,
\quad
\eps_i=\frac\eps{2C}\cdot i\cdot\pbra{\frac{(1+\alpha)(1-\beta)}{1+\beta}}^{i-1},
\quad
\delta_i=\frac\delta{2C}\cdot i\cdot\pbra{\frac{(1+\alpha)(1-\beta)}{1+\beta}}^{i-1}
$$\;
\Return{\Reduction{$\Tcal,\ell,\pbra{M_i,\eta_i,\alpha_i,\tau_i,\eps_i,\delta_i}_{i\in[\ell]}$}}
\end{algorithm}

By \Cref{lem:truncated_laplace_mechanism}, \textsf{Line 2} is $(\eps/2,\delta/2)$-DP.
By the proof of \Cref{thm:estimation}, \textsf{Lines 3--5} are $(\eps/2,\delta/2)$-DP.
Thus the privacy guarantee follows immediately from \Cref{fct:basic_composition}.

On the other hand, by the definition of truncated Laplace distribution, $M$ is always an upper bound of the weight of the root of $\Tcal$. 
Then the choice of the parameters follows exactly as in the proof of \Cref{thm:estimation}.
Together with the assumption on $\taumin$, the conditions in \Cref{thm:estimation} are satisfied, which implies the correctness of the algorithm.
\end{proof}

\subsection[Proof of Corollary 31]{Proof of \Cref{cor:estimation_lower_bound_rwalk}}\label{app:cor:estimation_lower_bound_rwalk}

\begin{proof}
Let $\Acal_1,\ldots,\Acal_s$ be independent copies of $\Acal$.
Define algorithm $\Bcal$ to output the entry-wise median of $\Acal_1,\ldots,\Acal_s$.
By \Cref{fct:basic_composition} and \Cref{fct:post-processing}, $\Bcal$ is $(s\cdot\eps,s\cdot\delta)$-DP.

Let $w_u$ be the true weight of node $u$ and let $\bw^{(i)}$ be the output of each $\Acal_i$.
Then the output of $\Bcal$, denoted by $\tbw$, is $\Median\pbra{\bw^{(1)},\ldots,\bw^{(s)}}$.
Then for any $u\in\nodes(\Tcal)$ we have
\begin{align*}
\Pr\sbra{\abs{\tilde w_u-w_u}>t_u}
&\le\Pr\sbra{\frac1s\sum_{i=1}^s\indicator_{\abs{w^{(i)}_u-w_u}>t_u}<\frac12}
\tag{$t_u=\alpha\cdot\max\cbra{w_u,\tau_u}$}\\
&\le\exp\pbra{-2\cdot s\cdot\pbra{\frac12-\eta}^2}
\tag{by a Chernoff--Hoeffding bound}\\
&\le\frac\kappa4<\frac18.
\end{align*}

Then applying \Cref{lem:estimation_lower_bound_rwalk} to $\Bcal$, we have the desired bound
\begin{equation*}
\taumax=\Omega\pbra{\frac{(1-\alpha)\cdot\pbra{d-3-(d-1)\cdot\Hcal(\kappa)}}{s\cdot\eps}}. 
\tag*{\qedhere}
\end{equation*}
\end{proof}

\subsection[Proof of Theorem 29]{Proof of \Cref{thm:add-err-lb}}\label{app:apx_lb}

Here we prove \Cref{thm:add-err-lb}: Any $(\eps, \delta)$-DP algorithm that allows only additive-only error (i.e., $\mRMSE$) must incur an error of $\Omega_{\eps, \delta}(\sqrt{d})$, as stated more formally below. 
This implies that the Gaussian mechanism (i.e., \Cref{cor:baseline_algorithms}) is essentially optimal for this setting.

To prove \Cref{thm:add-err-lb}, we start by giving a generic lower bound for $\bW$-linear queries. To state the lower bound, we need to define the $\gamma_2$-norm: 
For any matrix $\bW$, define 
$$
\gamma_2(\bW) := \min_{\bR^\top \bA = \bW} \vabs{\bR}_{\linfltwo} \vabs{\bA}_{\linfltwo}.
$$

The following theorem essentially follows from \cite{EdmondsNU20}, which asserts that $\gamma_2(\bW)$ captures the minimum $\mRMSE$ achievable by $(\eps,\delta)$-DP algorithms. 
Note that below we only state the lower bound; the upper bound follows from the so-called \emph{factorization mechanism}~\cite{NikolovTZ13}\footnote{The factorization mechanism is an instantiation of the matrix mechanism \cite{DBLP:journals/vldb/LiMHMR15}. See also~\cite{EdmondsNU20} where it is explained with notation more similar to ours.}.

\begin{theorem}\label{thm:linear-queries-characterization}
For any $\eps > 0$ and any $\delta > 0$ sufficiently small depending on $\eps$, there is a constant $C_{\eps, \delta} > 0$ that any $(\eps, \delta)$-DP algorithm $\cM$ for $\bW$-linear query must incur $\mRMSE$ at least $C_{\eps, \delta} \cdot \gamma_2(\bW)$.
\end{theorem}
\begin{proof}[Proof Sketch]
The proof is essentially the same as \cite[Theorem 30]{EdmondsNU20} which gives a similar characterization for the $\ell_2^2$-error. 
The only (non-trivial) change is that, in the proof of \cite[Theorem 28]{EdmondsNU20}, instead of considering $\vabs{\bR}_F$, we consider $\vabs{\bR}_{\linfltwo}$ and obtain
\begin{align*}
\vabs{\bR}_{\linfltwo} = \max_{i \in [m]} \sqrt{(\bR^\top \bR)_i} = \max_{i \in [m]} \sqrt{(\bR \bR^\top)_i} = \max_{i \in [m]} \sqrt{{\bm\Sigma}_{i, i}} = \mRMSE(\cM; \bW),
\end{align*}
where the second equality is due to the fact that $\bR = {\bm\Sigma}^{1/2}$ is symmetric.
\end{proof}

Let $\cT_\textsf{bin}$ be the complete binary tree of depth $d$, which has $m = 2^d - 1$ nodes and $n = 2^{d - 1}$ leaves. 
Recall that the tree aggregation problem can be viewed as the $\bW^{\cT_\textsf{bin}}$-linear query problem, where $\bW^{\cT_\textsf{bin}} \in \bin^{m \times n}$ is such that $W^{\cT_\textsf{bin}}_{i, j} = 1$ iff the node $i$ is an ancestor of (or itself) the leaf $j$. 
For each node $i$, we also define $\depth(i)$ to be the number of nodes on the path from root to $i$ minus one (e.g., $\depth(j) = d - 1$ for a leaf $j$).

To prove \Cref{thm:add-err-lb}, we need the following characterization of $\gamma_2$-norm based on the nuclear norm \cite{mathias1993hadamard}.
Here we adopt the presentation of \cite[Theorem 9]{LeeSS08}.
For a matrix $\bA$, let 
$\vabs{\bA}_*$ denote its \emph{nuclear norm}, i.e., the sum of the singular values of $\bA$.

\begin{theorem}[{\cite{mathias1993hadamard,LeeSS08}}] \label{thm:gamma2-dual}
For any matrix $\bW \in \R^{m \times n}$,
$$
\gamma_2(\bW) = \max_{\substack{\bu \in \R^n, \bv \in \R^m\\\vabs{\bu}_2, \vabs{\bv}_2 \leq 1}} \vabs{\bW \circ \bv\bu^\top}_*,
$$
where $\circ$ denotes the entrywise matrix product.
\end{theorem}

The above bound allows us to give a lower bound on $\gamma_2(\bW)$ by carefully choosing $\bu$ and $\bv$. 
We remark that this approach is similar to previous works, e.g., for range queries~\cite{rny033}.

\begin{proof}[Proof of \Cref{thm:add-err-lb}]
From \Cref{thm:linear-queries-characterization}, it suffices to show that $\gamma_2(\bW^{\cT_\textsf{bin}}) \geq \Omega(\sqrt{d})$. 

To prove this, we will use the dual characterization from \Cref{thm:gamma2-dual}. 
We select $\bv$ and $\bu$ as follows: 
Let $\bu \in \R^n$ be such that all entries have values $1/\sqrt{n}$ and let $\bv \in \R^m$ be such that $v_i = 1/\sqrt{{2^{\depth(i)} \cdot d}}$. It is simple to verify that $\vabs{\bu}_2, \vabs{\bv}_2 \leq 1$. Let $\bU = \bW^{\cT_\textsf{bin}} \circ \bv\bu^\top$. Consider $\bU^\top\bU$. 
By the definition of $\bW^{\cT_\textsf{bin}}, \bu, \bv$, we have
\begin{align*}
(\bU^\top\bU)_{j, k} = \frac{1}{n\cdot d} \sum_{\ell = 0}^{d-1} \frac{1}{2^{\ell}} \cdot \ind[\text{leaves } j, k \text{ share the same ancestor at depth } \ell].
\end{align*}

For convenience, let $\lambda=1/(n\cdot d)$.
We directly compute the eigenvectors of $\bU^\top\bU$:
\begin{itemize}
\item First is the all-ones vector, corresponding to eigenvalue $\lambda\cdot\left(2^{d-1} + \frac{2^{d-2}}{2^1} + \cdots + \frac{1}{2^{d-1}}\right) \geq \lambda \cdot n$.
\item The remaining eigenvectors are $\bz^i$'s for every internal node $i$.
Let $i_L$ and $i_R$ be the left and right child node of $i$ respectively, then $\bz^i$ is defined as
\begin{align*}
z^i_j =
\begin{cases}
1 &\text{ if } j \text{ is a descendant of } i_L, \\
-1 &\text{ if } j \text{ is a descendant of } i_R, \\
0 &\text{ otherwise.}
\end{cases}
\end{align*}
Let $\ell=\depth(i)$. This gives an eigenvalue of 
$$
\lambda\cdot\left(\frac{2^{d-1-(\ell+1)}}{2^{\ell+1}} + \frac{2^{d-1-(\ell+2)}}{2^{\ell+2}} + \cdots + \frac{1}{2^{d-1}}\right) \geq \lambda \cdot 2^{d-2\ell-3} = \lambda \cdot n / 2^{2(\ell+1)}.
$$
\end{itemize}
To summarize, $\bU^\top\bU$ has an eigenvalue at least $\lambda \cdot n$ and has $2^\ell$ eigenvalues at least $\lambda \cdot n / 2^{2(\ell+1)}$ for each $\ell=0,1,\ldots,d-2$.

Since the singular values of $\bU$ are simply the square roots of eigenvalues of $\bU^\top\bU$, we compute its nuclear norm as:
\begin{align*}
\vabs{\bU}_* \geq \sum_{\ell = 0}^{d - 2} 2^\ell\cdot\sqrt{\lambda \cdot n / 2^{2(\ell+1)}} 
= \Omega(d \cdot \sqrt{\lambda \cdot n})
= \Omega(\sqrt{d}).
\end{align*}
Then by \Cref{thm:gamma2-dual}, we have $\gamma_2(\bW^{\cT_\textsf{bin}}) \geq \Omega(\sqrt{d})$ as desired.
\end{proof}

\section[Smoothed Relative Error from alpha-Multiplicative RMSE]{Smoothed Relative Error from $\alpha$-Multiplicative RMSE}
\label{app:err-comp}

In this section, we show that an estimator with small $\alpha$-$\RMSE$ also has a small ``smoothed'' relative error, which is often used in empirical evaluations (e.g., \cite{QardajiYL13,ZhangXX16}).
Recall that the smoothed relative error can be defined as
\begin{align*}
\REL_{\kappa}(\tz, z) := \frac{\E_{\tz}\left[|\tz - z|\right]}{\max\{z, \kappa\}},
\end{align*}
where $\kappa > 0$ is the ``smoothing factor''.

The relationship between this notion of error and ours is stated below.

\begin{lemma}
For any $\alpha, \kappa > 0$ and any estimator $\tz$ of $z \in \R$, we have 
\begin{align*}
\REL_{\kappa}(\tz, z) \leq \sqrt{2} \cdot \left(\frac{\RMSE_\alpha(\tz, z)}{\kappa} + \alpha\right).
\end{align*}
\end{lemma}

\begin{proof}
We have
\begin{align*}
\E_{\tz}\left[|\tz - z|\right] &~\leq~ \sqrt{\E_{\tz}[|\tz - z|^2]} 
\tag{by convexity}\\
&~\leq~ \sqrt{\E_{\tz}[\pbra{\max\cbra{\abs{\tz - z} - \alpha\cdot z, 0} + \alpha \cdot z}^2]} \\
&~\leq~ \sqrt{\E_{\tz}[2\pbra{\max\cbra{\abs{\tz - z} - \alpha\cdot z, 0}}^2 + 2\pbra{\alpha \cdot z}^2]} 
\tag{since $(a+b)^2\le2a^2+2b^2$}\\
&~\leq~ \sqrt{\E_{\tz}[2\pbra{\max\cbra{\abs{\tz - z} - \alpha\cdot z, 0}}^2]} + \sqrt{2}\cdot(\alpha \cdot z)\\
&~=~ \sqrt{2} \cdot \RMSE_\alpha(\tz, z) + \sqrt{2}\cdot(\alpha \cdot z).
\end{align*}

Plugging this back to the definition of $\REL_{\kappa}$, we get
\begin{align*}
\REL_{\kappa}(\tz, z) \leq \sqrt{2} \cdot \left(\frac{\RMSE_\alpha(\tz, z) + \alpha\cdot z}{\max\{z, \kappa\}}\right)
\leq \sqrt{2} \cdot \left(\frac{\RMSE_\alpha(\tz, z)}{\kappa} + \alpha\right).
\tag*{\qedhere}
\end{align*}
\end{proof}

\end{document}